%% file: generalized_quantum_hydrodynamics.tex
\documentclass[%
nofootinbib,
superscriptaddress,
 preprint,
showpacs,preprintnumbers,
 amsmath,amssymb,
 aps,
 pra,
 longbibliography,
 floatfix,
 lengthcheck,%
]{revtex4-1}

\newcommand{\APPREF}[1]{Appendix~\ref{#1}}

\usepackage{generalized_quantum_hydrodynamics}

\newtheorem{theorem}{Theorem}
\newtheorem{lemma}{Lemma} 

\usepackage{xcolor}
\hypersetup{
	colorlinks=true,
	linkcolor=blue!50!red,
	urlcolor=red!70!black
} 

\begin{document}
\let\oldfootnote\footnote
\renewcommand{\footnote}[1]{\,\oldfootnote{#1}}

\input{generalized_quantum_hydrodynamics_body.tex}

\appendix

\input{generalized_quantum_hydrodynamics_app.tex}

\onecolumngrid
\bibliography{generalized_quantum_hydrodynamics}

\end{document}

%% file: generalized_quantum_hydrodynamics_body.tex
\title{Quantum and semiclassical dynamics as fluid theories where gauge matters}

\author{Dmitry V. Zhdanov}
\email{dm.zhdanov@gmail.com}
\affiliation{
Institute of Spectroscopy of the Russian Academy of Sciences, Moscow, 142190, Russia
}
\author{Denys I. Bondar}
\affiliation{
Tulane University, New Orleans, LA 70118, USA
}
\begin{abstract}
The family of trajectories-based approximations employed in computational quantum physics and chemistry is very diverse. For instance, Bohmian and Heller's frozen Gaussian semiclassical trajectories seem to have nothing in common. Based on a hydrodynamic analogy to quantum mechanics, we furnish the unified gauge theory of all such models. In the light of this theory, currently known methods are just a tip of the iceberg, and there exists an infinite family of yet unexplored trajectory-based approaches. Specifically, we show that each definition for a semiclassical trajectory corresponds to a specific hydrodynamic analogy, where a quantum system is mapped to an effective probability fluid in the phase space. We derive the continuity equation for the effective fluid representing dynamics of an arbitrary open bosonic many-body system. We show that unlike in conventional fluid, the flux of the effective fluid is defined up to Skodje’s gauge [R. T. Skodje \textit{et. al.} Phys. Rev. A \textbf{40}, 2894 (1989)]. We prove that the Wigner, Husimi and Bohmian representations of quantum mechanics are particular cases of our generic hydrodynamic analogy, and all the differences among them reduce to the gauge choice. Infinitely many gauges are possible, each leading to a distinct quantum hydrodynamic analogy and a definition for semiclassical trajectories. We propose a scheme for identifying practically useful gauges and apply it to improve a semiclassical initial value representation employed in quantum many-body simulations.
\end{abstract}
\maketitle


\section{Introduction}
\begin{figure}[bp]
\vspace{-10pt}
\captionsetup[subfloat]{position=top}

\def\spheight{0.175\textwidth}
\subfloat[``Default'' gauge $\GPH{=}0$ ($\wW_{p}{=}\wW_{x}{=}1/\sqrt 2$)
\label{@FIG:intro(Husimi)}
]{
\includegraphics[height=\spheight]{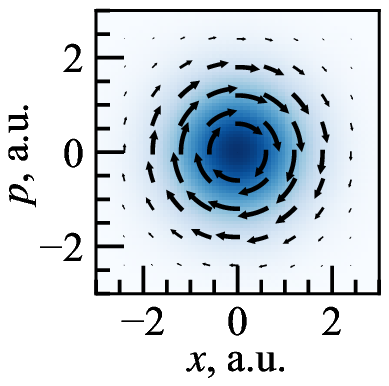}
}
\subfloat[Regularized Bohmian gauge
\label{@FIG:intro(RBG)}
]{\includegraphics[height=\spheight]{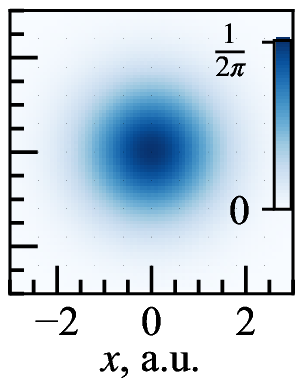}
}

\vspace{-32pt}\captionsetup[subfloat]{position=bottom}
\null\hfill
\hspace{-3pt}
\subfloat[Gauge $\GPH{=}\GPH_{\idx G}$ ($\wW_{p}{=}\wW_{x}{=}1/\sqrt 2$)
\label{@FIG:intro(GG)}
]{
\includegraphics[height=\spheight]{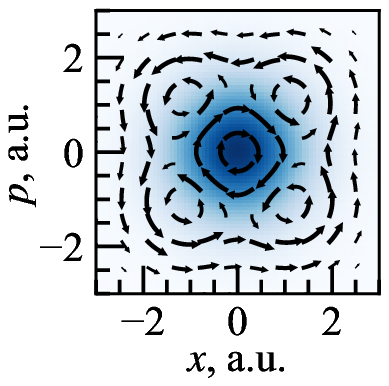}
}
\hspace{-6.5pt}
\subfloat[Bohmian gauge ($\wW_{p}/\wW_{x}{\to}\infty$)
\label{@FIG:intro(Bohm)}
]{
\includegraphics[height=\spheight]{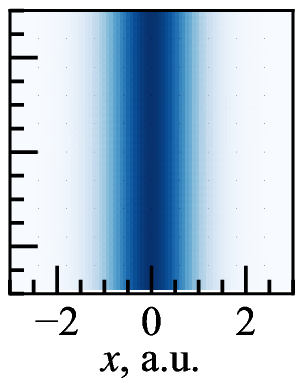}
}
\hfill\null
\caption{The fluid analogies for the ground state of a quantum harmonic oscillator with the Hamiltonian $\hat H{=}\frac12(\hat p^2{+}\hat x^2)$ $(\hbar{=}1)$: \protect\subref{@FIG:intro(Husimi)} the standard Husimi representation (Eqs.~\eqref{10.-J^H}); \protect\subref{@FIG:intro(RBG)} the regularized Bohmian gauge; \protect\subref{@FIG:intro(GG)} the gauge $\GPH{=}\GPH_{\idx G} =\pi \sum_{\delta x,\delta p{=}\pm1} \exp[-\frac34((x{+}\delta x)^2{+}(p{+}\delta p)^2)]$; \protect\subref{@FIG:intro(Bohm)} the Bohmian gauge (Eq.~\eqref{12+.-gen_bohmian_gauge_potential}). The probability fluid density is shown in blue. Black arrows show the vector field of fluid flow.
\label{@FIG:intro}}
\end{figure}
Fluid analogies for complex multidimensional quantum dynamics lay the basis for modern semiclassical computational methods \cite{2018-Weinbub,2019-Larder,1981-Brown,2001-Miller,2006-Saha,2007-Pollak,2015-Hele,2009-Berg, 2016-Vacher, 2001-Donoso, 2006-Lopez, 2008-Shalashilin} of many-body physics \cite{2017-Foss-Feig,2011-Chianca,2011-Cockburn,2015-Kordas}, chemistry \cite{1983-Carruthers,2016-Mai,2017-Orr,2013-Bonnet,2017-Makhov} and optics \cite{2003-OConnell}. In these analogies, the evolution of a quantum state is represented as a flow of an effective compressible probability fluid in the phase space (akin to classical statistical mechanics). Three fluid representations of quantum mechanics used are the Wigner \cite{BOOK-Zachos,1932-Wigner,1957-Stratonovich}, Husimi \cite{1940-Husimi} and Bohmian \cite{1927-Madelung,BOOK-Bacciagaluppi}. Each of them is exact and fully captures the effects of quantum nonlocality and quantum interference. This fact became universally accepted in 1950s after numerous and intriguing debates%
\footnote{See, e.g., the historic debates between Jos\'e Moyal and Paul Dirac on possibility to express quantum mechanics in terms of classical-valued phase space variables \cite{BOOK-Curtright}. A simple illustration of phase space representations of quantum superposition and entangled states is given in \APPREF{@APP:08-WF&HF-an_example}.}%
. Nowadays, the mathematical correspondence between quantum mechanics and classical hydrodynamics can even be demonstrated in actual experiments with fluids \cite{2019-Rozenman,2013-Dragoman}. More importantly, this very correspondence enables the description of quantum evolution in terms of phase-space trajectories -- the ultimate way to beat the curse of dimensionality in numerical applications. However, none of existing fluid analogies are fully developed. Namely, the exact form of fluid trajectories is unknown, except for one-dimensional systems \cite{2013-Veronez,2017-Kakofengitis} and a special class of multidimensional closed systems \cite{1989-Skodje}.

In this work, the general exact fluid analogy for an open multidimensional bosonic system is fully developed. The derived continuity equation for effective fluid has an ambiguity (pointed by Skodje \cite{1989-Skodje}) akin to the gauge invariance in electrodynamics. Consequently, the form of fluid trajectories is a matter of Skodje's flux gauge fixing, which has been overlooked in earlier studies \cite{2015-Colomes,2010-Hiley}. Our findings unify the fluid analogies to quantum mechanics. They reveal that the ``conventional'' Wigner, Husimi and Bohmian quantum hydrodynamic representations are just a tip of the iceberg and represent three out of infinitely many possible Skodje's gauge fixings. In particular, we show that the Bohmian mechanics (Fig.~\ref{@FIG:intro(Bohm)}) is nothing but a singular limiting case of the Husimi representation (Fig.~\ref{@FIG:intro(Husimi)}) in a specific gauge. As exemplified in Fig.~\ref{@FIG:intro}, different gauges result in strikingly different fluid analogies reflecting incompatible aspects of the wave-particle duality. For instance, the fluxes in the Husimi (Fig.~\ref{@FIG:intro(Husimi)}) and Bohmian (Fig.~\ref{@FIG:intro(Bohm)}) gauges highlight the non-vanishing zero-point energy and the stationarity of ground states, respectively.

Moreover, we prove by example that yet unexplored Skodje's gauges constitute a powerful resource to improve the accuracy of semiclassical numerical methods. We develop a methodology to screen for useful gauges and employ it to solve the Schr\"{o}dinger equation in the basis of time-dependent squeezed coherent states evolving along fluid trajectories. A comparison with benchmark initial value representations, such as the coupled coherent states (CCS) approach \cite{2008-Shalashilin}, confirms a superiority of our numerical method to capture tunneling dynamics.

The paper is organized as follows. The next section \ref{@APP:01} reviews the Wigner-Weyl formalism from dynamical perspective, which, to our knowledge, has not been systematically presented in literature but is critical for understanding our reasoning. Readers not interested in methodological details may skip Sec.~\ref{@APP:01} and go directly to Sec.~\ref{@SEC:Fluid_analogy}--\ref{@SEC:Bohmian_mechanics} containing our key results. Sec.~\ref{@SEC:Numerical_example} provides a simple numerical example demonstrating the superiority of the developed framework over traditional semiclassical initial value representations used in computational quantum physics and chemistry. A broader impact of our gauge analysis on quantum science and engineering is addressed in concluding remarks. Proofs of all the theorems and important technical details for applying the developed methods to real-world multidimensional problems are moved into the appendices.

Throughout the paper, we use bold symbols to denote vector quantities characterizing multidimensional systems. In particular, $\pp{=}\{p_1,...,p_{\dimensionality}\}$ and $\xx{=}\{x_1,...,x_{\dimensionality}\}$ denote the momentum and position coordinates of an $N$-dimensional system in the $2\dimensionality$-dimensional phase space, and $\ket{\xx}$ denotes a quantum position eigenstate in Dirac notations.

\section{Informal introduction to Wigner-Weyl quantization\label{@APP:01}}
The detailed expositions of the Wigner-Weyl formalism \cite{1932-Wigner,1946-Groenewold,1949-Bartlett,1957-Stratonovich} can be found in a variety of articles, tutorials and textbooks. Wigner's works \cite{1932-Wigner,1981-O-Connell,1984-Hillery} can serve as the physically appealing and intuitive introduction. Readers seeking for a more formal and axiomatic presentation might prefer the Stratonovich approach detailed in Refs.~\cite{1998-Brif,2011-Cahen}. For thorough discussions on the semiclassical limits, classical analogies and applications one can refer to Refs.~\cite{2010-Polkovnikov,BOOK-Zachos}.

This introduction to the Wigner-Weyl formalism is somewhat non-standard and largely informal. It is not intended to be a substitute of above-mentioned works. Rather, its objective is to introduce the ideas of the Wigner quantization from dynamical perspective in the spirit of our earlier works relying on the operational dynamical modelling (ODM) \cite{2012-Bondar,2015-Zhdanov}, which enables a smooth transition to the hydrodynamic interpretation of quantum mechanics.

It will be convenient for us to treat the general quantum Hermitian operators $\hat A{=}A(\hat{\pp},\hat{\xx})$ as symmetrized polynomials of the form 
\begin{gather}\label{app01.-Hermitian_operator_assumptions}
A(\hat{\pp},\hat{\xx}){=}\sum_{r_1,...,r_{2\dimensionality}}a_{r_1,r_2,...,r_{2\dimensionality}}\left\{\hat p_1^{r_1}...\hat p_{\dimensionality}^{r_{\dimensionality}}, \hat x_1^{r_{\dimensionality+1}}...\hat x_{\dimensionality}^{r_{2\dimensionality}}\right\}_{+},
\end{gather}
where $\left\{\odot,\odot \right\}_{+}$ stands for the anticommutator and $a_{r_1,r_2,...,r_{2\dimensionality}}$ are some real coefficients. This assumption is not really critical and just helps us to avoid extra complications by dealing with complex expansion coefficients.

\subsection{Key concepts in nutshell\label{@APP:01.-key_concepts}}

Quantum mechanics prescribes that the momentum and position operators $\hat{\pp}$ and $\hat{\xx}$ must satisfy identities
\begin{gather}\label{app01.-x,p_commutation_relations}
[\hat p_{n_1},\hat x_{n_2}]{=}{-}i\hbar\delta_{n_1,n_2},~~[\hat x_{n_1},\hat x_{n_2}]{=}[\hat p_{n_1},\hat p_{n_2}]{=}0.
\end{gather}
The way to represent these operators is solely up to us. The coordinate representation
\begin{subequations}\label{app01.-coordinate&momentum_representations}
	\begin{gather}\label{app01.-position_representation}
	\hat x_n{=}x_n,~~\hat p_n{=}{-}i\hbar\tpder{}{x_n},
	\end{gather}
	and momentum representation
	\begin{gather}\label{app01.-momentum_representation}
	\hat x_n{=}i\hbar\tpder{}{p_n},~~\hat p_n{=}p_n
	\end{gather}
\end{subequations}
are the most typical choices. However, nothing prevents us from choosing something more interesting
\begin{gather}\label{app01.-left_Bopp_px_operators}
\elBopp{x_n}{=}x_n{+}\lbdOp[2]{p_n},~~\elBopp{p_n}{=}p_n{-}\lbdOp[2]{x_n}.
\end{gather}
Here we replaced the common ``hat'' above the quantum mechanical operator with the right-pointing curved arrow. The reason will become clear shortly. The operators \eqref{app01.-left_Bopp_px_operators} are called the ``left'' Bopp operators. It is straightforward to verify that they satisfy the identities \eqref{app01.-x,p_commutation_relations}.

At the first glance, the representation \eqref{app01.-left_Bopp_px_operators} looks inconvenient: We doubled the number of variables for no reason. However, one already can notice the remarkable feature that the Bopp operators \eqref{app01.-left_Bopp_px_operators} reduce to the conventional phase space variables $x_n$ and $p_n$ in the classical limit $\hbar{\to}0$. This is clearly not the case for the representations \eqref{app01.-coordinate&momentum_representations}.

What kind of representation do the Bopp operators lead to? To answer this question, it is instructive first to review the position representation \eqref{app01.-position_representation} in more detail. Let $\hat{\rho}$ be a density matrix describing the generic state of a quantum system. Its position representation is $\rho(\xx',\xx''){=}\matel{\xx'}{\hat{\rho}}{\xx''}$. The generic term of the form $A(\hat{\pp},\hat{\xx})\hat{\rho}
B(\hat{\pp},\hat{\xx})$ is represented as
\begin{gather}
\hat A'\hat B''\rho(\xx',\xx''){=}\hat B''\hat A'\rho(\xx',\xx''),
\end{gather}
where $\hat A'{=}A({-}\lbdOp{\xx'},\xx')$ and $\hat B''{=}B(\lbdOp{\xx''},\xx'')$. Note the absence of the ``$-$'' sign in the momentum argument of $\hat B''$. This is the consequence of the fact that $\hat{\rho}\hat B{=}(\hat B^{\dagger}\hat{\rho}^{\dagger})^{\dagger}$. For this reason, the left-acting ``double-primed'' operators satisfy the commutation relations, which are complex conjugate of \eqref{app01.-x,p_commutation_relations} 
\begin{gather}\label{app01.-x,p_commutation_relations_right}
[\hat p_{n_1}'',\hat x_{n_2}'']{=}i\hbar\delta_{n_1,n_2}.
\end{gather}
In addition, the ``primed'' and ''double-primed'' operators act on different variables and hence commute
\begin{gather}\label{app01.-left_right_commute}
[\hat A',\hat B'']{=}0.
\end{gather}

What will happen if we switch to the Bopp representation \eqref{app01.-left_Bopp_px_operators}? What will be the analogs of ``double-primed'' operators? It appears that they can be expressed in terms of the following ``right'' Bopp operators:
\begin{gather}\label{app01.-right_Bopp_px_operators}
\erBopp{x_n}{=}x_n{-}\lbdOp[2]{p_n},~~\erBopp{p_n}{=}p_n{+}\lbdOp[2]{x_n},
\end{gather}
which will be marked with the left-pointing curved arrows. Indeed,
\begin{gather}\label{app01.-x,p_commutation_relations_Bopp_right}
[\erBopp{p_{n_1}},\erBopp{x_{n_2}}]{=}i\hbar\delta_{n_1,n_2}
\end{gather}
and
\begin{gather}\label{app01.-left_right_Bopp_commute}
[\elBopp{A},\erBopp{B}]{=}0
\end{gather}
for any $\elBopp{A}{=}A(\lBopp{\pp},\lBopp{\xx})$ and $\erBopp{B}{=}B(\rBopp{\pp},\rBopp{\xx})$. It is obvious that the relations \eqref{app01.-x,p_commutation_relations_Bopp_right} and \eqref{app01.-left_right_Bopp_commute} are identical to the equalities \eqref{app01.-x,p_commutation_relations_right} and \eqref{app01.-left_right_commute}%
\footnote{According to the prominent result of Leon Cohen, one can introduce infinitely many different phase space representations of quantum mechanics \cite{1966-Cohen}. However, the relations \eqref{app01.-x,p_commutation_relations_Bopp_right} and \eqref{app01.-left_right_Bopp_commute} make the Wigner representation special and unique.}.
Hence, the generic correspondence rule for the Bopp representation should read 
\begin{gather}\label{app01.-Bopp_correspondence_principle}
A(\hat{\pp},\hat{\xx})\hat{\rho}
B(\hat{\pp},\hat{\xx})\to A(\lBopp{\pp},\lBopp{\xx})B(\rBopp{\pp},\rBopp{\xx})\WF(\pp,\xx).
\end{gather}
Here $\WF(\pp,\xx)$ is the associated representation of the density matrix, which, as we are going to show right now, is exactly the Wigner function. Probably, the easiest way to deduce the form of $\WF(\pp,\xx)$ is to exploit the analogies with the coordinate and momentum representations \eqref{app01.-coordinate&momentum_representations} of a pure quantum state $\ket{\psi}$. Denote the wavefunctions in the position and momentum representations \eqref{app01.-position_representation} and \eqref{app01.-momentum_representation} as $\psi(\xx){=}\scpr{\xx}{\psi}$ and $\psi(\pp){=}\scpr{\pp}{\psi}$, respectively. Recall that the wavefunctions $\psi(\xx)$ and $\psi(\pp)$ are related by the Fourier transform $\FT{\xx}{\pp}$
\begin{gather}\label{app01.-Fourier_transform_operator}
\psi(\pp){=}\FT{\xx}{\pp}[\psi(\xx)]{\equiv}\tfrac{1}{(2\pi\hbar)^{\frac{N}2}}\inftyints\psi(\xx)e^{-i\frac{\pp\cdot\xx}{\hbar}}\diff^N\xx.
\end{gather}
Let us apply the inverse Fourier transform to $\WF(\pp,\xx)$ with respect to the momentum variables $\pp$: $\WF(\pp,\xx)\to\tilde\WF(\llambda,\xx){=}\IFT{\pp}{\llambda}[\WF(\pp,\xx)]$%
\footnote{The function $\WF(\llambda,\xx)$, known as the Blokhintsev function, was first introduced in Ref.~\cite{1940-Blokhintzev}.}. 
By comparing with \eqref{app01.-coordinate&momentum_representations}, one can see that the respective changes in the Bopp operators should be
\begin{gather}\notag
\lBopp{\xx}{\to}\xx{+}\tfrac{\llambda}2,~~\lBopp{\pp}{\to}\lbdOp{\llambda}{+}\lbdOp[2]{\xx},\\%
\rBopp{\xx}{\to}\xx{-}\tfrac{\llambda}2,~~\rBopp{\pp}{\to}\lbdOp{\llambda}{-}\lbdOp[2]{\xx}.\label{app01.-Blokhintsev_transform_for_Bopp_operators}
\end{gather}
Relations \eqref{app01.-Blokhintsev_transform_for_Bopp_operators} identify $\tilde\WF(\llambda,\xx)$ as the position representation of a density matrix with the additional variable substitutions $\xx'{\to}\xx{+}\tfrac{\llambda}2$ and $\xx''{\to}\xx{-}\tfrac{\llambda}2$, i.e.,
$
\tilde\WF(\llambda,\xx){\propto}\matel{\xx{+}\tfrac{\llambda}2}{\hat{\rho}}{\xx{-}\tfrac{\llambda}2}
$. The expression for $\WF(\pp,\xx)$ follows from the relation $\WF(\pp,\xx){=}\FT{\llambda}{\pp}[\tilde\WF(\llambda,\xx)]$
\begin{gather}\label{app01.-Wigner_function}
\WF(\pp,\xx){=}\tfrac{1}{(2\pi\hbar)^{\dimensionality}}\inftyints\matel{\xx{+}\tfrac{\llambda}2}{\hat{\rho}}{\xx{-}\tfrac{\llambda}2}e^{{-}i\frac{\pp\cdot\llambda}{\hbar}}\diff^{\dimensionality}\llambda.
\end{gather}
Eq.~\eqref{app01.-Wigner_function} coincides with the definition of Wigner function.%
\footnote{The normalization prefactor in Eq.~\texorpdfstring{\eqref{app01.-Wigner_function}}{} is chosen such that 
$\inftyints \mathrm d^{\dimensionality}\pp\, \mathrm d^{\dimensionality}\xx\WF(\pp,\xx){=}1$.
} 
In the case of a pure state $\ket{\psi}$ the definition \eqref{app01.-Wigner_function} reduces to 
\begin{gather}\label{app01.-WF_of_pure_state}
\WF(\pp,\xx){=}\tfrac1{\left(2\pi\hbar\right)^{\dimensionality}}\inftyints\psi^*(\xx{-}\tfrac{\xx'}2)\psi(\xx{+}\tfrac{\xx'}2)e^{-i\frac{\pp\cdot\xx'}{\hbar}}\diff^N\xx'.
\end{gather}
\subsection{Moyal product and Weyl symbols}
Consider an arbitrary Hermitian operator of the form 
\begin{gather}\label{app01.-separable_operator}
\hat O{=} O(\hat{\pp},\hat{\xx}){=}O_1(\hat{\xx}){+}O_2(\hat{\pp}).
\end{gather}
The associated Bopp operator $O(\lBopp{\pp},\lBopp{\xx})$ can be expanded into the Taylor series
\begin{align}
\lBopp{O}&{=}
O_1(\xx{+}\lbdOp[2]{\pp}){+}O_2(\pp{-}\lbdOp[2]{\xx}){=}\notag\\
&O_1(\xx)\sum_{l=0}^{\infty}\tfrac{\left(i\tfrac{\hbar}2(\pderl{\xx}{\cdot}\pderr{\pp})\right)^l}{l!}{+}O_2(\pp)\sum_{l=0}^{\infty}\tfrac{\left({-}i\tfrac{\hbar}2(\pderl{\pp}{\cdot}\pderr{\xx})\right)^l}{l!}{=}\notag\\
&O(\pp,\xx)\sum_{l=0}^{\infty}\tfrac{\left(i\tfrac{\hbar}2(\pderl{\xx}{\cdot}\pderr{\pp}{-}\pderl{\pp}{\cdot}\pderr{\xx})\right)^l}{l!}{=}
O(\pp,\xx)\star,
\end{align}
where we introduced the Moyal product
\begin{gather}\label{app01.-Moyal_product}
\star{=}\exp\big(i\tfrac{\hbar}2(\pderl{\xx}{\cdot}\pderr{\pp}-\pderl{\pp}{\cdot}\pderr{\xx})
\big).
\end{gather}
It can be shown in a similar fashion that $\rBopp{O}\WF(\pp,\xx){=}\WF(\pp,\xx){\star}O(\pp,\xx)$. These relations can be generalized to arbitrary operators and allow us to rewrite the correspondence relation \eqref{app01.-Bopp_correspondence_principle} in an appealing form
\begin{gather}\label{app01.-Moyal_correspondence_principle}
A(\hat{\pp},\hat{\xx})\hat{\rho}
B(\hat{\pp},\hat{\xx})\to \Weyl A(\pp,\xx){\star}\WF(\pp,\xx){\star}\Weyl B(\pp,\xx).\tag{\ref{app01.-Bopp_correspondence_principle}'}
\end{gather}
However, there is a complication that generally $\Weyl A(\pp,\xx){\ne}A(\pp,\xx)$, $\Weyl B(\pp,\xx){\ne}B(\pp,\xx)$ in the above formula, except for operators of the separable form \eqref{app01.-separable_operator}. Here we will not go deeper into the construction and properties of the quantities $\Weyl A$ and $\Weyl B$, called the Weyl symbols of operators, and instead refer readers to excellent review \cite{1989-Cohen}, where an interesting historical remarks can be also found.

\subsection{Quantum dynamics in Wigner representation}
Using the correspondence relations \eqref{app01.-Bopp_correspondence_principle} and \eqref{app01.-Moyal_correspondence_principle}, one can write the Liouville–von Neumann equation
\begin{gather}
\tpder{}{t}\hat{\rho}{=}\tfrac{1}{i\hbar}[\hat H,\hat{\rho}]
\end{gather}
in the Wigner representation as
\begin{gather}\label{app01.-quantum_Liouville_equation}
\tpder{}{t}\WF{=}\{\{\Weyl{H},\WF\}\},
\end{gather}
where the right hand side
\begin{gather}
\{\{\Weyl{H},\WF\}\}{=}\frac{1}{i\hbar}(\lBopp H{-}\rBopp H)\WF{=}\frac{1}{i\hbar}(\Weyl H{\star}\WF{-}\WF{\star}\Weyl H){=}\notag\\
\frac{2}{\hbar}\Weyl H\sin\big(\tfrac{\hbar}2(\pderl{\xx}{\cdot}\pderr{\pp}-\pderl{\pp}{\cdot}\pderr{\xx})\big)\WF\label{app01.-Moyal_bracket}
\end{gather}
is called the Moyal bracket of $\Weyl H$ and $\WF$. Eq.~\eqref{app01.-quantum_Liouville_equation} is often referred as the quantum Liouville equation. The reason for such a name becomes clear once the Moyal bracket \eqref{app01.-Moyal_bracket} is expanded into the series with respect to $\hbar$
\begin{gather}
\tpder{}{t}\WF{=}\{\Weyl{H},\WF\}{+}O(\hbar^2),
\end{gather}
where $\{\Weyl{H},\WF\}{=}\pder{\Weyl H}{\xx}\pder{\WF}{\pp}{-}\pder{\Weyl H}{\pp}\pder{\WF}{\xx}$ is the usual classical Poisson bracket. Thus, because of relations \eqref{app01.-left_Bopp_px_operators} the quantum Liouville equation \eqref{app01.-quantum_Liouville_equation} directly reduces to its expected classical counterpart in the limit $\hbar{\to}0$.

\section{Fluid analogy for bosonic systems\label{@SEC:Fluid_analogy}}
Consider a generic $\dimensionality$-dimensional open system with Hamiltonian $\hat H$ obeying the Lindblad-like master equation \cite{1972-Kossakowski,1976-Lindblad,2003-Havel}
\begin{gather}\label{10.-Lindblad_equation}
\tpder{}{t}\hat\rho{=}\tfrac{i}{\hbar}[\hat\rho,\hat H]{+}\tfrac12\textstyle{\sum_{k}}\big([\hat L_k(t){\hat\rho},\hat L^{\dagger}_k(t)]{+}\mbox{h.c.}\big),
\end{gather}
where dissipation operators $\hat L_k$ are generally time- and $\hat\rho$-dependent. We wish to cast Eq.~\eqref{10.-Lindblad_equation} for the density matrix $\hat\rho$ into an evolution equation for effective multidimensional probability fluid. Our starting point is the Wigner-Weyl formalism reviewed in the previous section, where the state of the system is described by the Wigner function \eqref{app01.-WF_of_pure_state}.
\begin{theorem}[see \APPREF{@APP:07-theorem_continuity_WF_proof} for proof]\label{10.-theorem_continuity_WF}
Master equation \eqref{10.-Lindblad_equation} can be cast into the continuity equation
\begin{gather}\label{10.-WF_continuity_equation}
\tpder{}t\WF(\pp,\xx){=}{-}\mathbb{\nabla}\cdot\WWC(\pp,\xx){=}{-}\tpder{}{\pp}{\cdot}\WWC_{\pp}{-}\tpder{}{\xx}{\cdot}\WWC_{\xx},
\end{gather}
where the components of the $2N$-dimensional flow are%
\begin{subequations}\label{10.-J^W}
	\begin{gather}
	\WWC_{\pp}{=}{-}\WF{\sincproduct}\tpder{\Weyl H}{\xx}{+}
	\tfrac {i\hbar}2\textstyle{\sum_k}\big((\Weyl{L}_k{\star}\WF){\sincproduct}\tpder{\Weyl{L}^*_k}{\xx}{+}\mbox{c.c.}
	\big),\label{10.-J^W_p}
	\\
	\WWC_{\xx}{=}\WF{\sincproduct}\tpder{\Weyl H}{\pp}{-}
	\tfrac {i\hbar}2\textstyle{\sum_k}\big((\Weyl{L}_k{\star}\WF){\sincproduct}\tpder{\Weyl{L}^*_k}{\pp}{+}\mbox{c.c.}
	\big)\label{10.-J^W_x}.
	\end{gather}
\end{subequations}
\end{theorem}
\noindent Here $\Weyl L_k(\pp,\xx)$ and $\Weyl H(\pp,\xx)$ are the Weyl symbols of the operators $\hat H$ and $\hat L_k$. The binary operations $\star$ and $\sincproduct$ are defined as
\begin{gather}\notag
\star{=}e^{i\frac{\hbar}2(\pderl{\xx}{\cdot}\pderr{\pp}-\pderl{\pp}{\cdot}\pderr{\xx})},~
\sincproduct{=}\sinc\big(\tfrac{\hbar}2(\pderl{\xx}{\cdot}\pderr{\pp}{-}\pderl{\pp}{\cdot}\pderr{\xx})
\big),
\end{gather}
where $\sinc(z){=}{\sin(z)}/{z}$ and the arrows indicate the directions of differentiation, e.g., $f(\pp,\xx)\pderl{\xx}g(\pp,\xx)=g(\pp,\xx)\pderr{\xx}f(\pp,\xx){=}\tpder{f(\pp,\xx)}{\xx}g(\pp,\xx)$ \cite{2010-Polkovnikov}.

The analogy between classical fluid density and $\WF(\pp,\xx)$ is incomplete because the latter typically has negative values \cite{2018-Oliva}. To fix this issue, we introduce \emph{the generalized Husimi function} $\HF(\pp,\xx)$ as a Gaussian convolution of $\WF(\pp,\xx)$:
\begin{gather}
\HF(\pp,\xx){=}\KernWD_{\wwW_{\pp},\wwW_{\xx}}\WF(\pp,\xx){=}\inftyints\WF(\pp',\xx')\times\notag\\
\KernW_{\wwW_{\pp},\wwW_{\xx}}(\pp{-}\pp',\xx{-}\xx') \diff^N\pp'\diff^N\xx',\label{10.-HF(gen)_definition}
\end{gather}
where $\wwW_{\pp}$ and $\wwW_{\xx}$ are arbitrary width parameters obeying the inequality $\wW_{p_n}\wW_{x_n}{>}\tfrac{\hbar}2$, and the kernel $\KernW_{\wwW_{\pp},\wwW_{\xx}}$ is%
\footnote{The convolution operator can also be written in the differential form 
$
\KernWD_{\wwW_{\pp},\wwW_{\xx}}{=}\prod_{n=1}^{\dimensionality}\exp({\frac12\wW_{p_n}^2\tpder{^2}{p_n^2}{+}\frac12\wW_{x_n}^2\tpder{^2}{x_n^2}})
$ (see, e.g., Ref.~\cite{2003-Ulmer}).
}

\begin{gather}\label{10.-Husimi_kernel}
\KernW_{\wwW_{\pp},\wwW_{\xx}}(\delta\pp,\delta\xx){=}
\textstyle{\prod_{n=1}^{\dimensionality}}
\frac{1}{2\pi\wW_{p_n}\wW_{x_n}}e^{{-}\frac{\delta p_n^2}{2\wW_{p_n}^2}{-}\frac{\delta x_n^2}{2\wW_{x_n}^2}}.
\end{gather}

Function \eqref{10.-HF(gen)_definition} is everywhere strictly positive. It reduces to the standard Husimi function (which can have zeros) in the limit $\wW_{p_n}\wW_{x_n}{\to}\tfrac{\hbar}2$%
\footnote{Note that the Husimi kernels \eqref{10.-Husimi_kernel} with $\wW_{p_k}\wW_{x_k}{=}\tfrac{\hbar}2$ and $\wW_{p_k}\wW_{x_k}{>}\tfrac{\hbar}2$ can be related to thermal states of a harmonic oscillator with zero and non-zero temperatures, respectively \cite{1949-Bartlett,1986-Dodonov}. It worth stressing that, unlike physical thermal averaging, the ``thermal'' Husimi transform remains invertible and does not lead to any loss of information contained in the original Wigner function.}:
\begin{gather}\label{10.-HF_definition}
\HF(\ppc,\xxc){=}\KernWD_{\frac{\hbar}{\sqrt2\wwc},\frac{\wwc}{\sqrt2}}\WF(\ppc,\xxc){=}\tfrac{\matel{\ppc,\xxc}{\hat\rho}{\ppc,\xxc}}{(2\pi\hbar)^N}{\geq}0,
\end{gather}
where $\ket{\ppc, \xxc}$ is the multidimensional squeezed coherent state localized at $\{\ppc, \xxc\}$:
\begin{align}\label{10.-CS_definition}
\scpr{\xx}{\ppc,\xxc}{=} \textstyle{\prod_{n{=}1}^{\dimensionality}}\pi^{{-}\frac12}\wc_n^{{-}\frac14}
e^{-\frac{(x_n{-}\xc_{n})^2}{2\wc_n^2}{+}\frac i{\hbar}\pc_{n}(x_n{-}\xc_{n})}.
\end{align}

\begin{theorem}[see \APPREF{@APP:09-theorem_continuity_HF_proof} for proof]\label{10.-theorem_continuity_HF}
	Master equation \eqref{10.-Lindblad_equation} can be cast into the continuity equation
	\begin{gather}\label{10.-HF_continuity_equation}
	\tpder{}t\HF(\pp,\xx){=}{-}\mathbb{\nabla}\cdot\HHC(\pp,\xx){=}{-}\tpder{}{\pp}{\cdot}\HHC_{\pp}{-}\tpder{}{\xx}{\cdot}\HHC_{\xx},
	\end{gather}
	for the generalized Husimi function $\HF(\pp,\xx)$, where the components of the Husimi flow vector field $\HHC$ are expressed via the Wigner flows \eqref{10.-J^W} as
\begin{gather}\label{10.-J^H}
\HHC_{\pp} {=} \KernWD_{\wwW_{\pp},\wwW_{\xx}}\WWC_{\pp},~~\HHC_{\xx} {=} \KernWD_{\wwW_{\pp},\wwW_{\xx}}\WWC_{\xx}.
\end{gather}
\end{theorem}
Theorem~\ref{10.-theorem_continuity_HF} and strict positivity of $\HF(\pp,\xx)$ allow to introduce the trajectories of elementary parcels of the effective Husimi fluid in the Lagrangian picture:
\begin{gather}\label{10.-Husimi_trajectories}
\tDer{}t\xx{=}\vvH_{\xx}(\pp,\xx),~~\tDer{}t\pp{=}\vvH_{\pp}(\pp,\xx),
\end{gather}
where the components of the velocity field are
$
\vvH_{\pp}(\pp,\xx) = {\HHC_{\pp}(\pp,\xx)}/{\HF(\pp,\xx)}$, and $\vvH_{\xx}(\pp,\xx) = {\HHC_{\xx}(\pp,\xx)}/{\HF(\pp,\xx)}.
$

\section{Gauge transformations\label{@SEC:Gauge_transformations}}
In classical fluid dynamics, equations governing mass density redistribution follow from underlying velocity fields. The situation is reverse in quantum hydrodynamics: The flow fields \eqref{10.-J^W} and \eqref{10.-J^H} were recovered from the master equation \eqref{10.-Lindblad_equation}. Such a recovery is not unique \cite{1989-Skodje}: The substitutions
\begin{gather}\label{11.-gauge_transformations}
\WWC\to\WWC{+}\delta\WWC,~~\HHC\to\HHC{+}\delta\HHC
\end{gather}
leave the quantum dynamics (i.e, Eqs.~\eqref{10.-WF_continuity_equation} and \eqref{10.-HF_continuity_equation}) intact if the auxiliary gauge fields $\delta\WWC$ and $\delta\HHC$ have zero divergences:  $\mathbb{\nabla}\cdot\delta\WWC{=}\mathbb{\nabla}\cdot\delta\HHC{=}0$. For instance, $\delta\HHC$ may represent invariable swirls in the phase space:
\begin{gather}
\delta\HC_{p_n}{=}{-}\tpder{}{x_n}\GPH_n(\pp,\xx),~ \delta\HC_{x_n}{=}\tpder{}{p_n}\GPH_n(\pp,\xx),\label{11.-gauge_flows_H}
\end{gather}
where the gauge potentials $\GPH_n(\pp,\xx)$ are arbitrary bounded twice continuously differentiable functions. 

This gauge freedom has been highlighted in, e.g., Refs.~\cite{1989-Skodje,2001-Donoso,2006-Lopez,2013-Veronez}, although has been utilized only once \cite{2012-Wang}. However, it has been overlooked that the price for the gauge freedom is the constraint
\begin{gather}\label{11.-boundary_conditions_on_flows}
\left.\WWC,\HHC
\right|_{\xx,\pp{\to}{\pm}\infty}{\to}0
\end{gather}
ensuring that there is no probability sources and drains at $|\pp|,|\xx|{=}{\pm}\infty$. For instance, the constraint \eqref{11.-boundary_conditions_on_flows} rules out uniform gauge flows $\delta\HHC{=}\mathrm{const}$ allowed by Eqs.~\eqref{11.-gauge_flows_H}.

It is noteworthy that the analytical expressions \eqref{10.-J^W} linearly depend on $\WF$ and hence satisfy Eq.~\eqref{11.-boundary_conditions_on_flows} for all localized quasiprobability densities. This feature opens opportunities to design dynamical models obeying desired constraints $\flowcnstr_r(\HHC){=}0$ $(r{=}1,...R)$ without explicitly solving the complicated boundary value problem \eqref{11.-boundary_conditions_on_flows}. Namely, one can start from Eqs.~\eqref{10.-J^W}, \eqref{10.-J^H} defining \emph{the ``default'' gauge} $\GPH_n{=}0$ and then seek new gauge potentials $\GPH_n(\pp,\xx)$ obeying the following constraints:%
\begin{subequations}\label{10.-constraints_on_gauge_potentials}
\begin{gather}\label{10.-constraints_on_gauge_potentials_b}
\left.\GPH_n,\mathbb{\nabla}\GPH_n\right|_{\xx,\pp{\to}{\pm}\infty}{\to}0,\\
\flowcnstr_r(\HHC{+}\delta\HHC(\GPH_1,...,\GPH_{\dimensionality})){=}0.\label{10.-constraints_on_gauge_potentials_u}
\end{gather}
\end{subequations}

Let us demonstrate this generic scheme on an example.

\section{Bohmian mechanics\label{@SEC:Bohmian_mechanics}}
Consider a closed $\dimensionality$-dimensional quantum system with the Hamiltonian
\begin{gather}\label{02'.-Hamiltonian}
\hat H = \textstyle{\sum_{n{=}1}^{\dimensionality}}\tfrac{1}{2 m_n}\hat p_n^2{+}V(\hat{\xx}),
\end{gather} whose Weyl symbol is $\Weyl H(\pp,\xx){=}\sum_{k{=}1}^{\dimensionality}\frac{p_k^2}{2 m_k}{+}V(\xx)$. We assume that the system is in a pure state described by the wavefunction $\psi(\xx,t){=}\scpr{\xx}{\psi(t)}$. Our aim is to construct the hydrodynamic analogy (marked by double primes $''$) to the dynamics of this system, where every $n$-th component ${\vGB_{x_n}}$ of the fluid velocity field \eqref{10.-Husimi_trajectories} is independent of the corresponding conjugate momentum $p_n$. The constraint  \eqref{10.-constraints_on_gauge_potentials_u} on ${\GBC_{\xx}}$ becomes
\begin{gather}\label{12.-generalized_bohmian_position_flow}
{\GBC_{x_n}}(\pp,\xx){=}\tfrac{\pBohmReg_n(p_1,...,p_{n-1},p_{n+1},...,p_{\dimensionality},\xx)}{m_n}\HF''(\pp,\xx).
\end{gather}
Here the velocities ${\vGB_{x_n}}$ are expressed in terms of yet-unknown $p_n$-independent functions $\pBohmReg_n$, which we will call the \emph{regularized Bohmian momenta} for the reasons below. Following our scheme, one first deduces the $x$-components of the gauge field from Eq.~\eqref{12.-generalized_bohmian_position_flow} using the definition \eqref{11.-gauge_transformations} and Eqs.~\eqref{10.-J^H},
\begin{gather}\label{12.-Bohmian_J_x}
{\delta\GBC_{x}}{=}\tfrac1{m_n}\left(\pBohmReg_n\HF''(\pp,\xx)-\KernWD_{\wwW_{\pp}'',\wwW_{\xx}''}(p_n\WF(\pp,\xx))\right).
\end{gather}
This allows to find the gauge potential ${\GPGB_n}(\pp,\xx)$ from the second of Eqs.~\eqref{11.-gauge_flows_H},
\begin{align}
{\GPGB_n}&(\pp,\xx){=}\linftyint{p_n}\delta{\GBC_{x}}\diff p{=}
\tfrac1{m_n}\big(\pBohmReg_n
\linftyint{p_n}\HF''(\pp,\xx)\diff p_n{-}\notag\\
&\linftyint{p_n}\KernWD_{\wwW_{\pp}'',\wwW_{\xx}''}\left(p_n\WF(\pp,\xx)\right)\diff p\big).
\label{12+.-gen_bohmian_gauge_potential}
\end{align}
The boundary constraints \eqref{10.-constraints_on_gauge_potentials_b} in the limit $p_n{\to}{+}\infty$ specialize the regularized Bohmian momenta $\pBohmReg_n$,
\begin{gather}\label{12.-generalized_Bohmian_momentum}
\pBohmReg_n
{=}\tfrac{\inftyint\diff p_n\KernWD_{\wwW_{\pp}'',\wwW_{\xx}''}\left(p_n\WF(\pp,\xx)\right)}{\inftyint
	\HF''(\pp,\xx)\diff p_n}.
\end{gather}
Now the gauge \eqref{12+.-gen_bohmian_gauge_potential} is fully specified, and we can compute all the components of the gauge flow field using Eqs.~\eqref{11.-gauge_flows_H}.

Of special interest is the limiting case (marked by single primes $'$) of the just-constructed hydrodynamic analogy, where
\begin{gather}\label{12.-Bohmian_limit}
\wW_{x}''{=}\wW_{x}'{\to}0,~~\wW_{p}''{=}\wW_{p}'{\to}\infty.
\end{gather}

\begin{theorem}[see \APPREF{@APP:03+} for proof]\label{12.-theorem_1D_Bohmian_transform}
	The evolution of the generalized Husimi fluid ${\HF}'(p,x){=}\left.\KernWD_{\wW'_{p},\wW'_{x}}\WF(\pp,\xx)\right.$ in the gauge \eqref{12+.-gen_bohmian_gauge_potential} in the one-dimensional case $\dimensionality{=}1$ and in the limit \eqref{12.-Bohmian_limit} is represented by the flow
	\begin{align}\label{12.-J^H_Bohm}
	{\BC_{x}}&{=}\tfrac{\pBohm(x)}m{\HF}'(p,x),~~
	{\BC_{p}}{=}{-}\tpder{(V(x){+}\Vb(x))}{x}{\HF}'(p,x),
	\end{align}
	where 
	$\pBohm(x){=}\left.\pBohmReg\right|_{\wW_{x}''{\to}0,\wW_{p}''{\to}\infty}{=}\Re\big[\tfrac{\hat p\psi(x)}{\psi(x)}\big]$ is the \emph{Bohmian momentum} and $\Vb(x){=}{-}\tfrac{\hbar^2}{2m{|\psi(x)|}}{\pder{^2}{x^2}|\psi(x)|}$ is the \emph{quantum potential} introduced in the Bohmian theory.
\end{theorem}

Theorem~\ref{12.-theorem_1D_Bohmian_transform} indicates that the phase space velocities
$
{\vvB}{=}{{\BBC}}/{Q'}{=}\{{-}\pder{}{x}(V(x){+}\Vb(x)),\tfrac{\pBohm(x)}m\}
$
are independent of $p$. Thence, it follows from  Eqs.~\eqref{10.-Husimi_trajectories} and \eqref{12.-generalized_Bohmian_momentum} that
\begin{gather}\label{12.-Bohmian_trajectories}
\tDer{}t\pBohm{=}{-}\tpder{}{x}(V(x){+}\Vb(x)),~~\tDer{}tx{=}\tfrac{\pBohm(x)}m.
\end{gather}
Furthermore, one can verify by direct substitution that
\begin{gather}\label{12.-Bohmian_wavefunction}
\psi(x){=}|\psi(x)|\exp({\tfrac i{\hbar}\linftyint{x}\pBohm(x)\diff x}).
\end{gather}
Equations \eqref{12.-Bohmian_trajectories} and \eqref{12.-Bohmian_wavefunction} fully specify the evolution of pure state $\ket{\psi}$ in terms of $|\psi(x)|{\propto}\sqrt{\inftyint{\HF}'(p,x)\diff p}$ and $\pBohm$. They reproduce the familiar equations of the Bohmian mechanics, the oldest quantum hydrodynamic formalism introduced by Madelung and de Broglie in 1927 \cite{1927-Madelung,BOOK-Bacciagaluppi}.

The physical meaning of the Bohmian gauge specialized by theorem \ref{12.-theorem_1D_Bohmian_transform} can be clarified on the example of the quantum harmonic oscillator in Fig.~\ref{@FIG:intro}. Panels \subref*{@FIG:intro(Husimi)} and \subref*{@FIG:intro(Bohm)} of Fig.~\ref{@FIG:intro} compare two fluid representations of the oscillator's ground state $\ket{\psi_0}$. In the ``default'' gauge (panel \subref*{@FIG:intro(Husimi)}), the quantum advection, shown by black arrows, reproduces classical dynamics%
\footnote{
The Husimi flows depicted in Fig.~\ref{@FIG:intro(Husimi)} correspond to the phase space velocities $\vH_x(p,x){=}\tfrac{\KernWD_{\wW_{p},\wW_{x}}(p\WF(p,x))}{\HF(p,x)}{=}
\tfrac12 p$, and $\vH_p(p,x)=\frac{\KernWD_{\wW_{p},\wW_{x}}({-}x\WF(p,x))}{\HF(p,x)}{=}
{-}\tfrac12 x$. Note that they differ by factor $\frac12$ compared to the case of a classical oscillator. This difference is merely due to the Gaussian blurring \eqref{10.-HF(gen)_definition}.
}.
In contrast, the flow fields vanish in the Bohmian gauge (panel~\subref*{@FIG:intro(Bohm)}) in line with the notion of the stationary state. 

The classical-like drift of the Husimi fluid (Fig.~\ref{@FIG:intro(Husimi)}) is an obstacle for simulating quantum tunneling effects as in the case of the double-well potential shown in Fig.~\ref{@FIG:traj}. Due to quantum advection in the ``default'' gauge of Fig.~\ref{@FIG:traj(Husimi)}, a dense sampling of the entire phase space is needed to capture a small fraction of trajectories crossing the barrier. In contrast, tunneling trajectories in the Bohmian gauge, Fig.~\ref{@FIG:traj(Bohm)}, originate from a localized region near the barrier permitting a local dense sampling. This very property of the Bohmian mechanics has caught the eyes of applied physicists and chemists \cite{2016-Gu,2019-Garashchuk}, leading to a recent remarkable progress in modelling complex many-body systems \cite{2019-Larder}. The Bohmian formalism found a variety of applications \cite{BOOK-Holland,BOOK-Oriols,2014-Benseny,BOOK-Sanz} including studies of molecular processes (e.g., reactive scattering and nonadiabatic vibronic dynamics \cite{2001-Wyatt,2013-Curchod,2017-Min}), elastic collisions \cite{2012-Efthymiopoulos}, and quantum measurements \cite{2004-Durr}. 
Parallels between the Bohmian mechanis and the classical dynamical effects in silicone oil have also been drawn \cite{2015-Milewski}.

However, the transition to the Bohmian gauge is 
not invertible for a generic mixed quantum state: The singular convolution $\KernWD_{\wW_{p}'\to\infty,\wW_{x}'\to0}$ in Eq.~\eqref{10.-HF(gen)_definition} irrecoverably destroys the momentum information in the original Wigner function. Also, the Bohmian fluid spreads over an infinite phase space area (see Fig.~\ref{@FIG:traj(Bohm)}). The associated natural coherent states \eqref{10.-CS_definition}, following from the relation \eqref{10.-HF_definition}, collapse to a set of $\delta$-functions with undefined parameters $\pc_k$. Such a singular basic is not suitable for numerical simulations.

In order to circumvent all these issues, we propose \emph{the regularized Bohmian transformation}, where the gauge \eqref{12+.-gen_bohmian_gauge_potential} is applied to the localized Husimi function $\HF''{=}\KernWD_{\wwW''_{\pp},\wwW''_{\xx}}\WF$ constructed using finite width parameters $\wwW''{\ne}0,\infty$. This gauge inherits the aforementioned advantages of the Bohmian mechanics -- vanishing flows for steady states (Fig.~\ref{@FIG:intro(RBG)}) and the localized origins of tunneling trajectories in the phase space (Fig.~\ref{@FIG:traj(RBG)}). Moreover, compared to Fig.~\ref{@FIG:traj(Bohm)}, the resulting effective fluid trajectories in Fig.~\ref{@FIG:traj(RBG)} are smooth and completely free of the $\delta$-like singularities.

\begin{figure}[htp]
\def\spheight{0.175\textwidth}
\subfloat[``Default'' gauge $\GPH{=}0$ ($\wW_{p}{=}\wW_{x}{=}1/\sqrt 2$)
\label{@FIG:traj(Husimi)}
]{
\includegraphics[height=\spheight]{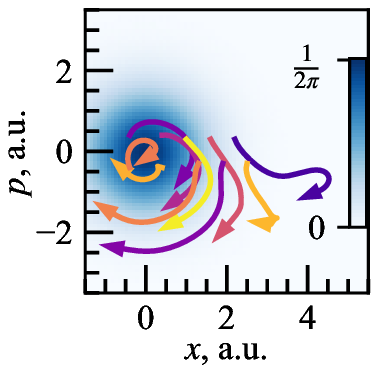}
}
\subfloat[Bohmian gauge ($\wW_{p}/\wW_{x}{\to}\infty$)
\label{@FIG:traj(Bohm)}
]{
\includegraphics[height=\spheight]{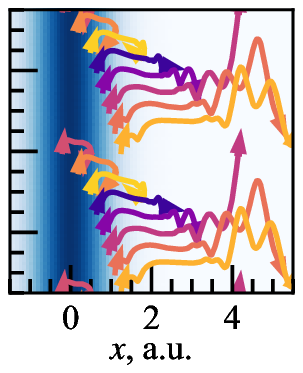}
}
\subfloat[Regularized Bohmian gauge ($\wW_{p}{=}\wW_{x}{=}1/\sqrt 2$)
\label{@FIG:traj(RBG)}
]{
\includegraphics[height=\spheight]{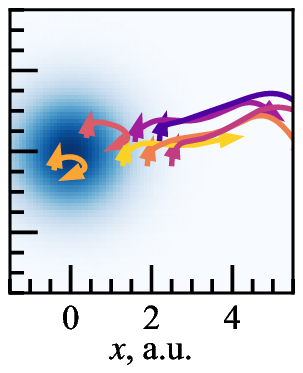}
}
\caption{The three fluid representations of the qunatum dynamics in the double-well potential $\hat V^{\idx{dw}}{=}\frac{1}{50} (\hat x{-}5)^2 \hat x^2$: \protect\subref{@FIG:traj(Husimi)} the ``default'' Husimi gauge (Eqs.~\eqref{10.-J^W} and \eqref{10.-J^H}); \protect\subref{@FIG:traj(Bohm)} the Bohmian gauge; \protect\subref{@FIG:traj(RBG)} the regualrized Bohmian gauge. The Husimi function $\HF$ of the initial state $\left.\scpr{x}{\psi}\right|_{t{=}0}{=}\sqrt[4]{\pi}e^{-x^2/2}$ is in blue. Colored curves depict the trajectories of effective fluid parcels.
\label{@FIG:traj}}
\end{figure}

\section{Numerical example\label{@SEC:Numerical_example}}
During last two decades, notable progress has been achieved in simulating quantum dynamics directly in the Wigner \cite{2015-Cabrera,2016-Bondar} or Husimi representations, especially thanks to the Martens group
\cite{2001-Donoso,2006-Lopez,2009-Wang,2012-Wang}. However, such simulations are inherently challenging. A particular issue is vanished quantum interference (and hence, the phase information) in the Husimi picture (see \APPREF{@APP:08-WF&HF-an_example} for details and an example). We argue that this issue can be turned into a numerical advantage when combined with the CCS- and FMS-type quantum chemistry methods \cite{2016-Vacher} to solve the Schr\"odinger equation $i\hbar\tpder{}t\ket{\psi}{=}\hat H\ket{\psi}$ via the anzatz
\begin{gather}\label{01.-cs-anzatz}
\ket{\psi(t)}{=}\textstyle{\sum_{k=1}^{\basissize}}\ac_k(t)
\scs{\ppc_k(t),\xxc_k(t)}~~~(\ac_k{\in}\complexes),
\end{gather} 
constructed from the time-dependent squeezed coherent states \eqref{10.-CS_definition}. It is shown in \APPREF{@APP:02} that the trajectories $\{\ppc_k(t),\xxc_k(t)\}$ of the centers of the basis states in the CCS method semiclassically approximate the Wigner flow lines defined by Eqs.~\eqref{10.-J^W}. Thus, we propose to evolve $\{\ppc_k(t),\xxc_k(t)\}$ with the flow of the generalized Husimi fluid $\HF{=}\KernWD_{\frac{\hbar}{\sqrt2{\wwc}},\frac{\wwc}{\sqrt2}}\WF$ defined by Eqs.~\eqref{10.-Husimi_trajectories}: 
\begin{gather}\label{13.-dpc/dt,dxc/dt}
\tder{}t\ppc_k{=}\vvH_{\pp}(\ppc_k,\xxc_k),~ \tder{}t\xxc_k{=}\vvH_{\xx}(\ppc_k,\xxc_k),~\vvH{=}\tfrac{\HHC}{\HF}.
\end{gather}
Here the width parameters $\wwW_{\pp}{=}\alpha{\hbar}/({\sqrt2\wwc})$ and 
$\wwW_{\xx}{=}\alpha{\wwc}/{\sqrt2}$ are chosen to match the spread of the basis states \eqref{10.-CS_definition} (see Eq.~\eqref{10.-HF_definition}), and $\alpha{=}1.1{>}1$ is chosen to ensure that $Q$ has no zeros. Importantly, in \APPREF{@APP:04} we reduce the rhs of Eqs.~\eqref{13.-dpc/dt,dxc/dt} to computationally cheap closed-form expressions applicable to generic multidimensional quantum systems with Hamiltonian \eqref{02'.-Hamiltonian}. Further numerical details can be found in \APPREF{@APP:06}.

As a specific test, we compute the solution to the one-dimensional Schr\"{o}dinger equation in a challenging case of tunneling through a high and narrow potential barrier accompanied by multiple scattering. The quantum Hamiltonian and initial wavefunction $\psi(x,t{=}0)$ in the dimensionless units $m{=}\hbar{=}1$ read
\begin{gather}\label{12.-model_problem_tunneling}
\hat H{=}\tfrac{\hat p^2}2{+}\tfrac{\hat x^2}2{+}25e^{-\left(\frac{\hat{x}}{0.35}\right)^2},~~
\psi(x,t{=}0)\propto e^{{-}\frac{(x-4)^2}2}.
\end{gather}
The results, summarized in Fig.~\ref{@FIG:W02}, show that the accuracy of simulating weak tunneling effects improves significantly if  the ``default'' Husimi gauge $\GPH{=}0$ (green curves) is replaced by the regularized Bohmian gauge (blue curves). Moreover, the latter gauge also delivers substantially better long-term estimates for both the shape of the quantum state and the tunneling probabilities than the standard CCS method \cite{2008-Shalashilin} (black curves).

\begin{figure}[htp]
\centering
\includegraphics[width=1\columnwidth]{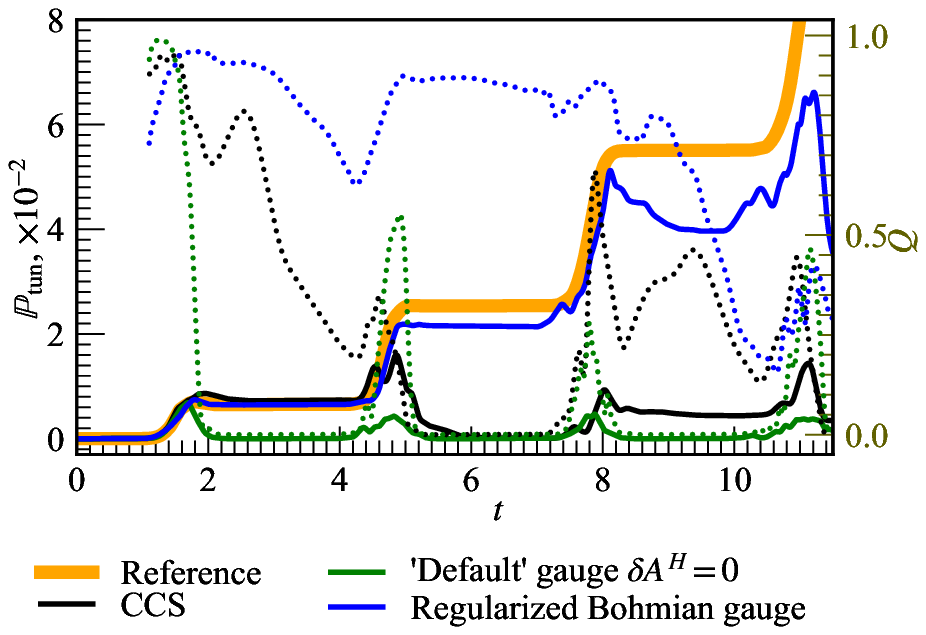}
\caption{The time-dependent tunneling probabilities $\tunprob{=}\matel{\psi}{h(\hat x)}{\psi}$ (solid curves, the left scale) for the problem \eqref{12.-model_problem_tunneling}, where $h(x)$ is the Heaviside step function. Also shown (dotted curves, the right scale) is the quality factor $Q{=}\frac{\left|\matel{\psiso}{h(\hat x)}{\psi}\right|^2}{\matel{\psiso}{h(\hat x)}{\psiso}\matel{\psi}{h(\hat x)}{\psi}}$, which relates the tunneled fraction of wavefunction $\psi(t)$ to the exact reference solution $\psiso(t)$ (orange) obtained via the standard split operator technique \cite{1983-Kosloff} on a grid. Identical initial bases of 21 coherent states \eqref{10.-CS_definition} are used in all the simulations (for details see \APPREF{@APP:06}). The link to source code can be found in Ref.~\cite{SOURCE_CODE}.\label{@FIG:W02}} 
\end{figure}

\section{Outlook}
The demonstrated utility of our theory calls for revisiting other applications of phase space methods in quantum mechanics and beyond. For instance, an open question is the optimal Skodje's gauge for examining quantum processes via recently proposed inventive use of the dynamical systems theory \cite{2013-Steuernagel,2013-Mason}. Exploring the fluid analogies to other initial value representations, such as G-MCTDH \cite{1999-Burghardt,BOOK-G-MCTDH}, could help to improve their numerical stability. The close ties between the phase space approaches to signal processing and quantum mechanics (see, e.g., Refs.~\cite{1989-Cohen,1966-Cohen,2011-Rojas}) make the developed methodology transferable to the time-frequency analysis of propagating signals in engineering applications. Connections with the gauge-dependent trajectory-based methods relying on the Voronoi tessellation \cite{2010-Coffey,2011-Coffey} also deserve a detailed exploration. Finally, we believe that our results blur the boundary between classical and quantum worlds by shedding new light on the controversial interpretation of trajectories in quantum mechanics \cite{BOOK-Penrose,2005-Leggett} and offer new ways to construct quantum-classical hybrid models \cite{2019-Bondar,2019-Gay-Balmaz}. 

\begin{acknowledgments}
~D.~Zh.~thanks Francisco Gonz\'alez Montoya for fruitful discussions and relevant advises. D.I.B. was supported by AFOSR (grant FA9550-16-1-0254), ARO (grant W911NF-19-1- 0377), and DARPA (grant D19AP00043). The views and conclusions contained in this document are those of the authors and should not be interpreted as representing the official policies, either expressed or implied, of AFOSR, ARO, DARPA, or the U.S. Government. The U.S. Government is authorized to reproduce and distribute reprints for Government purposes notwithstanding any copyright notation herein.
\end{acknowledgments}

%% file: generalized_quantum_hydrodynamics_app.tex
\section{Proof of theorem~\texorpdfstring{\ref{10.-theorem_continuity_WF}}{}\label{@APP:07-theorem_continuity_WF_proof}} 
Our starting point is the Liouville-von Neumann equation~\eqref{app01.-quantum_Liouville_equation}. Unlike the classical Liouville equation, Eq.~\eqref{app01.-quantum_Liouville_equation} does not preserve the phase space volume. However, it preserves the normalization \begin{gather}\label{app01.-WF_normalization}
\inftyints\diff^{\dimensionality}\pp\diff^{\dimensionality}\xx\WF(\pp,\xx){=}\Tr[\hat{\rho}]{=}1.
\end{gather} (The latter identity can be verified using Eq.~\eqref{app01.-Wigner_function}.) Hence, it should be possible to cast Eq.~\eqref{app01.-quantum_Liouville_equation} into the form of a continuity equation
\begin{gather}\label{app01.-WF_continuity_equation}
\tpder{}t\WF={-}\mathbb{\nabla}\cdot\WWC,
\end{gather}
where $\mathbb{\nabla}{=}\{\pder{}{\pp},\pder{}{\xx}\}$ and $\WWC{=}\{\WWC_{\pp},\WWC_{\xx}\}$ are the Wigner flows. In order to find  $\WWC_{\pp}$ and $\WWC_{\xx}$, we will need the following lemma:
\begin{lemma}\label{app07.-lemma_Moyal_bracket}
The Moyal bracket \eqref{app01.-Moyal_bracket} of two Weyl symbols $\Weyl A$ and $\Weyl B$ can be written as
\begin{gather}
\{\{\Weyl A,\Weyl B\}\}{=}\tpder{}{\xx}(\Weyl A\sincproduct\tpder{\Weyl B}{\pp}){-}\tpder{}{\pp}(\Weyl A\sincproduct\tpder{\Weyl B}{\xx}),
\end{gather}
where $\sincproduct$ denotes the binary operation 
\begin{gather}\label{app07.-sinc_product}
\sincproduct{=}\sinc\big(\tfrac{\hbar}2(\pderl{\xx}{\cdot}\pderr{\pp}-\pderl{\pp}{\cdot}\pderr{\xx})
\big).
\end{gather}
\end{lemma}
\begin{proof}
\begin{align*}
\tpder{}{\xx}&(\Weyl A\sincproduct\tpder{\Weyl B}{\pp}){-}\tpder{}{\pp}(\Weyl A\sincproduct\tpder{\Weyl B}{\xx}){=}\\
&(\Weyl A(\pderl{\xx}{+}\cancel{\pderr{\xx}})\sincproduct\tpder{\Weyl B}{\pp}){-}(\Weyl A(\pderl{\pp}{+}\cancel{\pderr{\pp}})\sincproduct\tpder{\Weyl B}{\xx}){=}\\
&\Weyl A\left(\pderl{\xx}\pderr{\pp}{-}\pderl{\pp}\pderr{\xx}\right)\sincproduct\Weyl B{=}\{\{\Weyl A,\Weyl B\}\}.
\end{align*}
\end{proof}

\begin{proof}[Proof of theorem \ref{10.-theorem_continuity_WF}]
	Using the correspondence rule \eqref{app01.-Moyal_correspondence_principle} and the definition \eqref{app01.-Moyal_bracket} of the Moyal bracket, one can convert a master equation 
	\begin{gather*}\label{app07.-Lindblad_equation}
	\tpder{}{t}\hat\rho{=}\tfrac{i}{\hbar}[\hat\rho,\hat H]{+}\tfrac12\sum_{k}\big([\hat L_k(t){\hat\rho},\hat L^{\dagger}_k(t)]{+}\mbox{h.c.}\big)
	\end{gather*}
	\noindent into the Wigner representation as
	\begin{gather}\label{app07.-Lindblad_equation_WF}
	\tpder{}{t}\WF{=}\{\{\Weyl{H},\WF\}\}{-}\tfrac{i\hbar}2\sum_{k}\big(\{\{\Weyl L_k(t){\star}\WF,\Weyl L^{*}_k(t)\}\}{+}\mbox{c.c.}\big).
	\end{gather}
	Using Lemma~\ref{app07.-lemma_Moyal_bracket}, Eq.~\eqref{app07.-Lindblad_equation_WF} can be further rewritten as
	\begin{align}
	\tpder{}{t}\WF{=}&{-}\tpder{}{\xx}(\WF\sincproduct\tpder{\Weyl H}{\pp}){+}\tpder{}{\pp}(\WF\sincproduct\tpder{\Weyl H}{\xx}){+}\notag\\
	&\tfrac{i\hbar}2\sum_{k}\bigg(
	\tpder{}{\xx}\big((\Weyl L_k(t){\star}\WF)\sincproduct\tpder{\Weyl L^{*}_k(t)}{\pp}\big){-}\notag\\
	&\tpder{}{\pp}\big((\Weyl L_k(t){\star}\WF)\sincproduct\tpder{\Weyl L^{*}_k(t)}{\xx}\big)
	{+}\mbox{c.c.}\bigg){=}\notag\\
	&{-}\tpder{}{\pp}{\cdot}\WWC_{\pp}{-}\tpder{}{\xx}{\cdot}\WWC_{\xx},
	\end{align}
	where finally
	\begin{subequations}\label{app01.-J^W}
		\begin{align}
		\WWC_{\pp}{=}&{-}\WF\sincproduct\tpder{\Weyl H}{\xx}{+}\tfrac{i\hbar}2\sum_{k}\bigg(\big((\Weyl L_k{\star}\WF)\sincproduct\tpder{\Weyl L^{*}_k}{\xx}\big){+}\notag\\
		&\big(\tpder{\Weyl L_k}{\xx}{\sincproduct}(\WF\star\Weyl L^{*}_k)\big)\bigg),\label{app01.-J^W_p}
		\end{align}
		\begin{align}
		\WWC_{\xx}{=}&\WF\sincproduct\tpder{\Weyl H}{\pp}{-}\tfrac{i\hbar}2\sum_{k}\bigg(\big((\Weyl L_k{\star}\WF)\sincproduct\tpder{\Weyl L^{*}_k}{\pp}\big){+}\notag\\
		&\big(\tpder{\Weyl L_k}{\pp}{\sincproduct}(\WF\star\Weyl L^{*}_k)\big)\bigg).\label{app01.-J^W_x}
		\end{align}
	\end{subequations}
\end{proof}

Let us make a remark about the important special case of a closed system with the Hamiltonian
\begin{gather}\label{app01.-separable_Hamiltonian}
\hat H{=}H(\hat{\pp},\hat{\xx}) = \sum_n\tfrac{\hat p_n^2}{2 m_n}{+}V(\hat{\xx}).
\end{gather}
This Hamiltonian \eqref{app01.-separable_Hamiltonian} has the separable form \eqref{app01.-separable_operator}, so that $\Weyl H{=}H(\pp,\xx)$. By substituting this Weyl symbol into Eqs.~\eqref{app01.-J^W} and integrating by parts an appropriate number of times, one can show that the state-averaged Wigner flows $\midop{\WWC}{=}\inftyints\WWC(\pp,\xx)\diff^{\dimensionality}\pp\diff^{\dimensionality}\xx$ coincide with the respective classical expressions:
\begin{subequations}\label{app01.-<J^W>}
\begin{gather}
\midop{\WWC_{\pp}(\WF)}{=}{-}\inftyints\tpder{V(\xx)}{\xx}\WF(\pp,\xx)\diff^{\dimensionality}\pp\diff^{\dimensionality}\xx\\
\midop{\WWC_{x_n}(\WF)}{=}\inftyints\tfrac{p_n}{m_n}\WF(\pp,\xx)\diff^{\dimensionality}\pp\diff^{\dimensionality}\xx.
\end{gather}
\end{subequations}

\section{Proof of theorem~\texorpdfstring{\ref{10.-theorem_continuity_HF}}{}\label{@APP:09-theorem_continuity_HF_proof}} 
The normalization condition \eqref{app01.-WF_normalization} holds for the generalized Husimi function as well. Indeed, one can check that the kernel $\KernW_{\wwW_{\pp},\wwW_{\xx}}$ defined by Eq.~\eqref{10.-Husimi_kernel} 
 satisfies the property
\begin{gather}
\inftyints\KernW_{\wwW_{\pp},\wwW_{\xx}}(\pp{-}\pp',\xx{-}\xx')\diff^{\dimensionality}\pp\diff^{\dimensionality}\xx{=}1.
\end{gather}
Using this identity, one finds that
\begin{align}
\inftyints&\HF(\pp,\xx)\diff^{\dimensionality}\pp\diff^{\dimensionality}\xx{=}\notag\\
&\inftyints\KernWD_{\wwW_{\pp},\wwW_{\xx}}{\cdot}\WF(\pp,\xx)\diff^{\dimensionality}\pp\diff^{\dimensionality}\xx{=}\notag\\
&\inftyints\WF(\pp',\xx')\diff^{\dimensionality}\pp'\diff^{\dimensionality}\xx'{=}1.\label{app01.-HF_normalization}
\end{align}
Identity~\eqref{app01.-HF_normalization} implies that the generalized Husimi function should satisfy a continuity-like equation similar to Eq.~\eqref{app01.-WF_continuity_equation}. 
\begin{proof}[Proof of theorem~\ref{10.-theorem_continuity_HF}]
The explicit form of the continuity equation for the generalized Husimi function can be deduced by applying the convolution operator $\KernWD_{\wwW_{\pp},\wwW_{\xx}}$ to the both sides of Eq.~\eqref{app01.-WF_continuity_equation}:
\begin{subequations}\label{app09.-WH_continuity_equation}
\begin{align}%
\KernWD_{\wwW_{\pp},\wwW_{\xx}}\left(\tpder{}{t}\WF\right)&{=}\tpder{}{t}\HF\\\KernWD_{\wwW_{\pp},\wwW_{\xx}}\left({-}\mathbb{\nabla}\cdot\WWC\right)&{=}{-}\mathbb{\nabla}{\cdot}\left(\KernWD_{\wwW_{\pp},\wwW_{\xx}}\WWC\right),
\end{align}
\end{subequations}
where the Wigner currents are defined by Eqs.~\eqref{app01.-J^W}. In the last equality we used the fact that the convolution operator $\KernWD_{\wwW_{\pp},\wwW_{\xx}}$ can be equivalently represented a in differential form as
\begin{align}
\KernWD_{\wwW_{\pp},\wwW_{\xx}}{=}\textstyle{\prod_{n=1}^{\dimensionality}}\exp({\frac12\wW_{p_n}^2\tpder{^2}{p_n^2}{+}\frac12\wW_{x_n}^2\tpder{^2}{x_n^2}}),
\label{app09.-Husimi_transform_diff}
\end{align}
from which the identity $\KernWD_{\wwW_{\pp},\wwW_{\xx}}\mathbb{\nabla}{=}\mathbb{\nabla}\KernWD_{\wwW_{\pp},\wwW_{\xx}}$ follows. By equating the right hand side of Eqs.~\eqref{app09.-WH_continuity_equation}, one finds that
\begin{gather}\label{app09.-J^H}
\HHC_{\pp} {=} \KernWD_{\wwW_{\pp},\wwW_{\xx}}\WWC_{\pp},~~\HHC_{\xx} {=} \KernWD_{\wwW_{\pp},\wwW_{\xx}}\WWC_{\xx}.
\end{gather}
\end{proof}
For completeness, let us also re-express the right hand side of Eqs.~\eqref{app09.-J^H} for the Husimi currents solely in terms of the Husimi function $\HF(\pp,\xx)$ using the identity
\begin{gather}
\KernWD_{\wwW_{\pp},\wwW_{\xx}}^{-1}{=}\KernWD_{-i\wwW_{\pp},i\wwW_{\xx}},
\end{gather}
which directly follows from Eq.~\eqref{app09.-Husimi_transform_diff} and allows to formally write
\begin{gather}\label{app09.-inverse_husimi_transform}
\WF(\pp,\xx){=}\KernWD_{-i\wwW_{\pp},i\wwW_{\xx}}\HF(\pp,\xx).
\end{gather}
Substitution of Eq.~\eqref{app09.-inverse_husimi_transform} into Eqs.~\eqref{app09.-J^H} gives
\begin{subequations}\label{app09.-J^H-spec}
	\begin{align}\label{app09.-J^H_p-spec}
	\HHC_{\pp}{=}&\left(e^{\ad_{\sum_{n=1}^{\dimensionality}\frac12\wW_{p_n}^2\pder{^2}{p_n^2}{+}\frac12\wW_{x_n}^2\pder{^2}{x_n^2}}}\tpder{\Weyl H(\pp,\xx)}{\xx}\right)
	{\sincproduct}
	\HF(\pp,\xx){=}\notag\\
	&
	\HF(\pp,\xx)\sincproduct^{\idx{H}}\tpder{}{\xx}\left(\KernWD_{\wwW_{\pp},\wwW_{\xx}}\Weyl H(\pp,\xx)\right),
	\end{align}
	\begin{align}
	\label{app09.-J^H_x-spec}
	\HHC_{\xx}{=}&{-}\left(e^{\ad_{\sum_{n=1}^{\dimensionality}\frac12\wW_{p_n}^2\pder{^2}{p_n^2}{+}\frac12\wW_{x_n}^2\pder{^2}{x_n^2}}}\tpder{\Weyl H(\pp,\xx)}{\pp}\right)
	{\sincproduct}
	\HF(\pp,\xx){=}\notag\\
	&{-}
	\HF(\pp,\xx)\sincproduct^{\idx{H}}\tpder{}{\pp}\left(\KernWD_{\wwW_{\pp},\wwW_{\xx}}\Weyl H(\pp,\xx)\right),
	\end{align}
\end{subequations}
where the adjoint mapping notation $\ad_{\hat X}\hat Y{=}[\hat X,\hat Y]$ is used and the symbol $\sincproduct^{\idx{H}}$ denotes the compound binary operation
\begin{gather}
\sincproduct^{\idx{H}}{=}\exp\left({\sum_{n{=}1}^{\dimensionality}\left(\wW_{x_n}^2\pderl{x_n}\pderr{x_n}{+}\wW_{p_n}^2\pderl{p_n}\pderr{pn}\right)}\right)\sincproduct.
\end{gather}

\section{Proof of theorem~\texorpdfstring{\ref{12.-theorem_1D_Bohmian_transform}}{}\label{@APP:03+}}
As a preliminary step, let us specialize the definition of the gauge potential 
\eqref{12+.-gen_bohmian_gauge_potential} for the case of a pure quantum state $\ket{\psi}$. For this, we will need the following property of the convolution operator $\KernWD_{\wwW_{\pp},\wwW_{\xx}}$:
\begin{gather}\label{app03+.-lemma_integration_convolution}
\linftyint{p_n}\KernWD_{\wwW_{\pp},\wwW_{\xx}}Z(\pp,\xx)\diff p_n{=}
\KernWD_{\wwW_{\pp},\wwW_{\xx}}\left(\linftyint{p_n}Z(\pp,\xx)\diff p_n\right),
\end{gather}
which is valid for any function $Z(\pp,\xx)$ exponentially small at $|\pp|\to\infty$ and/or $|\xx|\to\infty$. With the help of Eqs.~\eqref{app03+.-lemma_integration_convolution}, \eqref{app01.-WF_of_pure_state} and the equality $p_ne^{-i\frac{\pp\cdot\xx'}{\hbar}}{=}i\hbar\tpder{}{x_n'}e^{-i\frac{\pp\cdot\xx'}{\hbar}}$ we obtain
\begin{widetext}
\begin{align}
{\GPGB_n}&(\pp,\xx){=}\linftyint{p_n}\delta{\GBC_{x}}\diff p_n{=}
\tfrac1{m_n}\big(\pBohmReg_n
\linftyint{p_n}\HF''(\pp,\xx)\diff p_n{-}\linftyint{p_n}\KernWD_{\wwW_{\pp}'',\wwW_{\xx}''}\left(p_n\WF(\pp,\xx)\right)\diff p_n\big){=}\notag\\
&\tfrac1{m_n}\big(\pBohmReg_n
\KernWD_{\wwW_{\pp}'',\wwW_{\xx}''}\linftyint{p_n}\WF(\pp,\xx)\diff p_n{-}\KernWD_{\wwW_{\pp}'',\wwW_{\xx}''}\linftyint{p_n}\left(p_n\WF(\pp,\xx)\right)\diff p_n\big){=}\notag\\
&\tfrac1{m_n}\big(\pBohmReg_n\KernWD_{\wwW_{\pp}'',\wwW_{\xx}''}\linftyint{p}\tfrac{1}{(2\pi\hbar)^{\dimensionality}}\inftyint\psi^{\ast}(\xx{-}\tfrac{\xx'}2)\psi(\xx{+}\tfrac {\xx'}2)e^{-i\frac{\pp\cdot\xx'}{\hbar}}\diff^{\dimensionality}\xx'\diff p_n{-}\notag\\
&~~~~\KernWD_{\wwW_{\pp}'',\wwW_{\xx}''}\linftyint{p_n}p_n\tfrac{1}{(2\pi\hbar)^{\dimensionality}}\inftyint\psi^{\ast}(\xx{-}\tfrac{\xx'}2)\psi(\xx{+}\tfrac {\xx'}2)e^{-i\frac{\pp\cdot\xx'}{\hbar}}\diff^{\dimensionality}\xx'\diff p_n\big){=}\notag\\
&\tfrac1{m_n(2\pi\hbar)^{\dimensionality}}\bigg(\pBohmReg_n\KernWD_{\wwW_{\pp}'',\wwW_{\xx}''}\left(\mbox{P.V.}\inftyint\psi^{\ast}(\xx{-}\tfrac{\xx'}2)\psi(\xx{+}\tfrac {\xx'}2)\tfrac{i\hbar e^{-i\frac{\pp\cdot\xx'}{\hbar}}}{x_n'}\diff^{\dimensionality}\xx'\right){+}\notag\\
&~~~~\KernWD_{\wwW_{\pp}'',\wwW_{\xx}''}\left(\mbox{P.V.}\inftyint i\hbar\tpder{}{x_n'}\left(\psi^{\ast}(\xx{-}\tfrac{\xx'}2)\psi(\xx{+}\tfrac {\xx'}2)\right)\tfrac{i\hbar e^{-i\frac{\pp\cdot\xx'}{\hbar}}}{x_n'}\diff^{\dimensionality}\xx'\right)\bigg).\label{app03a.-Bohmian-Wigner_gauge_potential_nD}
\end{align}
\end{widetext}
The last equality in \eqref{app03a.-Bohmian-Wigner_gauge_potential_nD} was obtained via the integration by parts with respect to $x_n$.
\begin{proof}[Proof of theorem~\ref{12.-theorem_1D_Bohmian_transform}]
We are now interested in the one-dimensional case $\dimensionality{=}1$ 
when the Husimi tranform is performed in the singular limit $\wW''_{p}{=}\wW'_{p}{\to}\infty$, $\wW''_{x}{=}\wW'_{x}{\to}0$%
\footnote{Hereafter we will drop the dimension subscript $n{=}1$ to simplify notations.}%
. Note that the kernel $\KernW_{\wW_{p}'',\wW_{x}''}$ defined by
 Eq.~\eqref{10.-Husimi_kernel}
 in this limit reduces to 
\begin{align}\label{app02.-Husimi_kernel_in_Bohmian_limit}
\KernW_{\wW'_{p},\wW'_{x}}(p{-}p',x{-}x'){\simeq}\tfrac{\delta(x{-}x')}{\sqrt{2\pi}\wW'_{p}}{\simeq}\delta(x-x')\tfrac{\HF'(p,x)}{|\psi(x)|^2}.
\end{align}
After substituting Eq.~\eqref{app02.-Husimi_kernel_in_Bohmian_limit} into  Eq.~\eqref{12.-generalized_Bohmian_momentum}
 one finds that 
\begin{gather}\label{app03+.-Bohmian_momentum}
\left.\pBohmReg\right|_{\wW''_{p}{\to}\infty,\wW''_{x}{\to}0}{=}\frac{\inftyint p\WF(p,x)\diff p}{\inftyint\WF(p,x)\diff p}{=}\pBohm(x),
\end{gather}
where $\pBohm(x){=}\frac{\psi(x)^{\ast}\overleftrightarrow{p}\psi(x)}{|\psi(x)|^2}$ is the conventional Bohmian momentum, here $\overleftrightarrow{p}{=}\frac{i\hbar}2(\pderl{x}{-}\pderr{x})$. 

Eqs.~\eqref{app02.-Husimi_kernel_in_Bohmian_limit} and \eqref{app03+.-Bohmian_momentum} allow to cast the gauge potential~\eqref{app03a.-Bohmian-Wigner_gauge_potential_nD} into the form
\begin{align}
{\GPB}(p,x)&{=}\tfrac{\HF'(p,x)}{|\psi(x)|^2}\mbox{P.V.}\inftyint\tfrac{i \hbar\delta(y)}{ m y}\times\notag\\%
&\left(\pBohm(x){+}i\hbar\tpder{}{y}\right)\psi^{\ast}(x{-}\tfrac y2)\psi(x{+}\tfrac y2)\diff y.
\end{align}
Further simplifications can be made by expanding $\psi^{\ast}(x{-}\tfrac y2)$ and $\psi(x{+}\tfrac y2)$ into the series of $y$:
\begin{align}\label{app01.-Bohmian_gauge_potential_simplification}
{\GPB}&(p,x){=}\tfrac{\HF'(p,x)}{|\psi(x)|^2}\mbox{P.V.}\inftyint\tfrac{i \hbar\delta(y)}{ m y}
\left(\pBohm(x){+}i\hbar\tpder{}{y}\right)\times\notag\\%
&\bigg\{\psi^{\ast}(x)\psi(x){+}
\tfrac12y\big(\psi^{\ast}(x)\tpder{\psi(x)}{x}{-}\tpder{\psi^{\ast}(x)}{x}\psi(x)\big){+}\notag\\%
&\tfrac18y^2\big(\psi^{\ast}(x)\tpder{^2\psi(x)}{x^2}{-}2\tpder{\psi^{\ast}(x)}{x}\tpder{\psi(x)}{x}{+}\tpder{^2\psi^{\ast}(x)}{x^2}\psi(x)\big){+}\notag\\
&o(y^2)
\bigg\}\diff y.
\end{align}
The terms in the curly brackets proportional to $y^r$ with $r{>}2$ do not contribute to the integral over $y$. The rest of Eq.~\eqref{app01.-Bohmian_gauge_potential_simplification} can be split into even and odd terms with respect to $y$; the odd terms also do not contribute to the integral. By collecting the even terms and performing the integration, one obtains
\begin{align}
{\GPB}&(p,x){=}\tfrac{\HF'(p,x)}{|\psi(x)|^2}
\tfrac{i\hbar}{2m}\bigg\{
\pBohm(x)\big(\psi^{\ast}(x)\tpder{\psi(x)}{x}{-}\tpder{\psi^{\ast}(x)}{x}\psi(x)\big){+}\notag\\%
&\tfrac{i\hbar}{2}\big(\psi^{\ast}(x)\tpder{^2\psi(x)}{x^2}{-}2\tpder{\psi^{\ast}(x)}{x}\tpder{\psi(x)}{x}{+}\tpder{^2\psi^{\ast}(x)}{x^2}\psi(x)\big)\bigg\}.\label{app02.-gauge-A_unprettified}
\end{align}
Using the definition of the Bohmian momentum, one can simplify Eq.~\eqref{app02.-gauge-A_unprettified}
\begin{gather}\label{app02.-Bohmian-Husimi_gauge_potential}
{\GPB}(p,x){=}\tfrac{\hbar^2}{m}{{\HF}'(p,x)}\left(\tfrac{\psi(x)^{\ast}\overleftrightarrow{p}^2\psi(x)}{|\psi(x)|^2}-\pBohm(x)^2\right).
\end{gather}

Now we are ready to compute the fluid flow vector field corresponding to the Bohmian gauge. Its momentum component reads
\begin{align}\label{app02.-Bohmian_J^H_p-definition}
{\BC_{p}}&{=}{-}\tpder{}{x}{\GPB}(p,x){+}\left.\KernWD_{\frac{\hbar}{\sqrt2\wc},\frac{\wc}{\sqrt2}}\WC_{p}(p,x)\right|_{\wc{\to}0},
\end{align}
where $\WC_{p}(p,x)$ is given by Eq.~\eqref{app01.-J^W_p}. Evaluation of the second term in Eq.~\eqref{app02.-Bohmian_J^H_p-definition} using the specific form \eqref{app02.-Husimi_kernel_in_Bohmian_limit} of the convolution kernel $\KernW_{\wW_{p}',\wW_{x}'}$ gives
\begin{align}
\KernWD&_{\frac{\hbar}{\sqrt2\wc},\frac{\wc}{\sqrt2}}\left.\WC_{p}(p,x)\right|_{\wc{\to}0}
{=}{-}\tfrac{\HF'(p,x)}{|\psi(x)|^2}\times\notag\\%
&\inftyint\WF(p,x)\sinc\big(\tfrac{\hbar}2(\pderl{x}\pderr{p}-\pderl{p}\pderr{x})\big)\tpder{V(x)}x\diff p{=}\notag\\
&{-}\tfrac{\HF'(p,x)}{|\psi(x)|^2}\inftyint\WF(p,x)\sinc\big(\tfrac{\hbar}2\pderl{p}\pderr{x}\big)\tpder{V(x)}x\diff p{=}\notag\\
&{-}\tfrac{\HF'(p,x)}{|\psi(x)|^2}\inftyint\WF(p,x)\tpder{V(x)}x\diff p{+}\tfrac{\hbar}2\tfrac{\HF'(p,x)}{|\psi(x)|^2}\times\notag\\
&\xcancel{\inftyint\tpder{}{p}\bigg(\WF(p,x)\sum_{k{=}1}^{\infty}\tfrac{\left(-\tfrac{\hbar}2\pderl{p}\pderr{x}\right)^{2k-1}}{(2k+1)!}\tpder{^2V(x)}{x^2}\bigg)\diff p}{=}\notag\\
&{-}\HF'(p,x)\tpder{V(x)}x
.
\end{align}
The first term in Eq.~\eqref{app02.-Bohmian_J^H_p-definition} can be transformed as
\begin{align}
&\tpder{}{x}{\GPB}(p,x){=}\notag\\
&\tfrac{\hbar^2}{m}\tfrac{{\HF}'(p,x)}{|\psi(x)|^2}\tpder{}{x}\left(|\psi(x)|^2\left(\tfrac{\psi(x)^{\ast}\overleftrightarrow{p}^2\psi(x)}{|\psi(x)|^2}-\pBohm(x)^2\right)\right){=}\notag\\
&\tpder{\Vb(x)}{x}{\HF}'(p,x),
\end{align}
where $\Vb(x)$ is the quantum potential defined in theorem~\ref{12.-theorem_1D_Bohmian_transform}.

Thus, Eq.~\eqref{app02.-Bohmian_J^H_p-definition} can be rewritten as
\begin{subequations}\label{app02.-Bohmian_J^H}
	\begin{align}\label{app02.-Bohmian_J^H_p}
	{\BC_{p}}&{=}{-}\HF'(p,x)\tpder{}x\left(V(x)+\Vb(x)\right).
\end{align}
The positional component of the flow can be obtained in a similar fashion using the definition of ${\BC_x}$ (see Eq.~\eqref{12.-generalized_bohmian_position_flow}
):
\begin{align}
	{\BC_{x}}&{=}\tfrac{\pBohm}m{\HF}'(p,x)
\end{align}
\end{subequations}
Equalities \eqref{app02.-Bohmian_J^H} prove the theorem.
\end{proof}

\section{Quantum superposition and entangled states in Wigner and Husimi pictures\label{@APP:08-WF&HF-an_example}}
The possibility of the Wigner function $\WF(\pp,\xx,t)$ to take negative values compromises its interpretation as a conventional probability distribution (this is why it is often referred as the \emph{quasiprobability distribution}). The emergence of negative values is more the rule than the exception. For instance, Fig.~\ref{@FIG:W01a} illustrates the Wigner function $\WF_{\psi}$ of a superposition 
\begin{gather}
\ket{\psi}{\propto}\ket{\pc_a,\xc_a}{+}\ket{\pc_b,\xc_b}.
\end{gather}
of two squeezed coherent states defined in Eq.~\eqref{01.-cs-anzatz}
. One can see that $\WF_{\psi}$ consists of three components
\begin{gather}\label{app08.-superposition_state}
\WF_{\psi}(p,x){=}P_{aa}(p,x){+}P_{bb}(p,x)+2\Re[P_{ab}(p,x)],
\end{gather}
where
\begin{gather}
P_{ab}(p,x){\propto}
\tfrac1{2\pi\hbar}\inftyint\scpr{\pc_a,\xc_a}{x{-}\tfrac{x'}2}\scpr{x{+}\tfrac{x'}2}{\pc_b,\xc_b}e^{-i\frac{px'}{\hbar}}\diff x'.
\end{gather}
The first two terms in Eq.~\eqref{app08.-superposition_state} manifest themselves in Fig.~\ref{@FIG:W01a} by two Gaussian blobs around the points $\{p_a,x_a\}$ and $\{p_b,x_b\}$. These blobs are identical to the classical probability distribution for a particle having equal likelihoods of being localized either near $\{p_a,x_a\}$ or near $\{p_b,x_b\}$. However, the last term in Eq.~\eqref{app08.-superposition_state} adds non-classical sign-changing interference fringes between these blobs. It is because of this term the phase space velocity field $\WWC(\pp,\xx)/\WF(\pp,\xx)$ is both physically and mathematically ill-defined and singular. Consequently, one cannot unambiguously define the trajectories of the elementary parcels of quantum fluid represented by the time-dependent Wigner function \cite{2018-Oliva}. Furthermore, in a general case, the overlapping fringes form a complicated landscape, which is difficult to approximate numerically.

Panels \subref{@FIG:W01a} and \subref{@FIG:W01b} of Fig.~\ref{@FIG:W01} help to compare the Wigner and  Husimi representations for the same quantum superposition state. As expected, the Husimi function $\HF(p,x)$ is everywhere non-negative. At the same time, one can see that the interference fringes are nearly entirely smeared out in the Husimi representation by the Gaussian convolution (they are nearly two orders of magnitude smaller than the dominant peaks). Consequently, the quantum superpositions with the opposite phases shown in panels \subref{@FIG:W01b} and \subref{@FIG:W01c} look nearly identical.

Similar conclusions apply to an entangled state of two quantum particles, such as
\begin{gather}\label{app08.-entangled_state}
\ket{\idx{ent}}{\propto}\ket{\pc_{1,a},\xc_{1,a}}\ket{\pc_{2,a},\xc_{2,a}}{+}\ket{\pc_{1,b},\xc_{2,b}}\ket{\pc_{1,b},\xc_{2,b}},
\end{gather}
where indices 1 and 2 enumerate the particles. The corresponding Wigner function $\WF_{\idx{ent}}$ again consists of three terms
\begin{align}
\WF_{\idx{ent}}&(\pp,\xx){=}P_{aa}(p_1,x_1)P_{aa}(p_2,x_2){+}P_{bb}(p_1,x_1)P_{bb}(p_2,x_2){+}\notag\\
&2\Re[P_{ab}(p_1,x_1)P_{ab}(p_2,x_2)].\label{app08.-WF_of_entangled_state}
\end{align}
The first two terms describe classical correlations between particles' positions and momenta. The last term in Eq.~\eqref{app08.-WF_of_entangled_state} is composed of the product of terms $P_{ab}$, which, as we have learned from Fig.~\ref{@FIG:W01a}, are non-positive and oscillatory. This additional, markedly non-classical correlation term is the phase space signature of quantum nonlocality. We refer interested readers to Ref.~\cite{1991-Venugopalan} for a thorough analysis of the quantum nonlocality (expressed as a violation of the Bell inequality) in terms of the parameters of the Wigner function. 

These illustrations show that simulating quantum dynamics directly in the Wigner and Husimi representations is challenging. Nevertheless, it is worth noting the innovative methods developed in the Martens group, which made such simulations possible for low-dimensional systems \cite{2001-Donoso,2006-Lopez,2009-Wang,2012-Wang}.

\begin{figure*}[htp]
\centering
\subfloat[
]{
\includegraphics[width=0.25\textwidth]{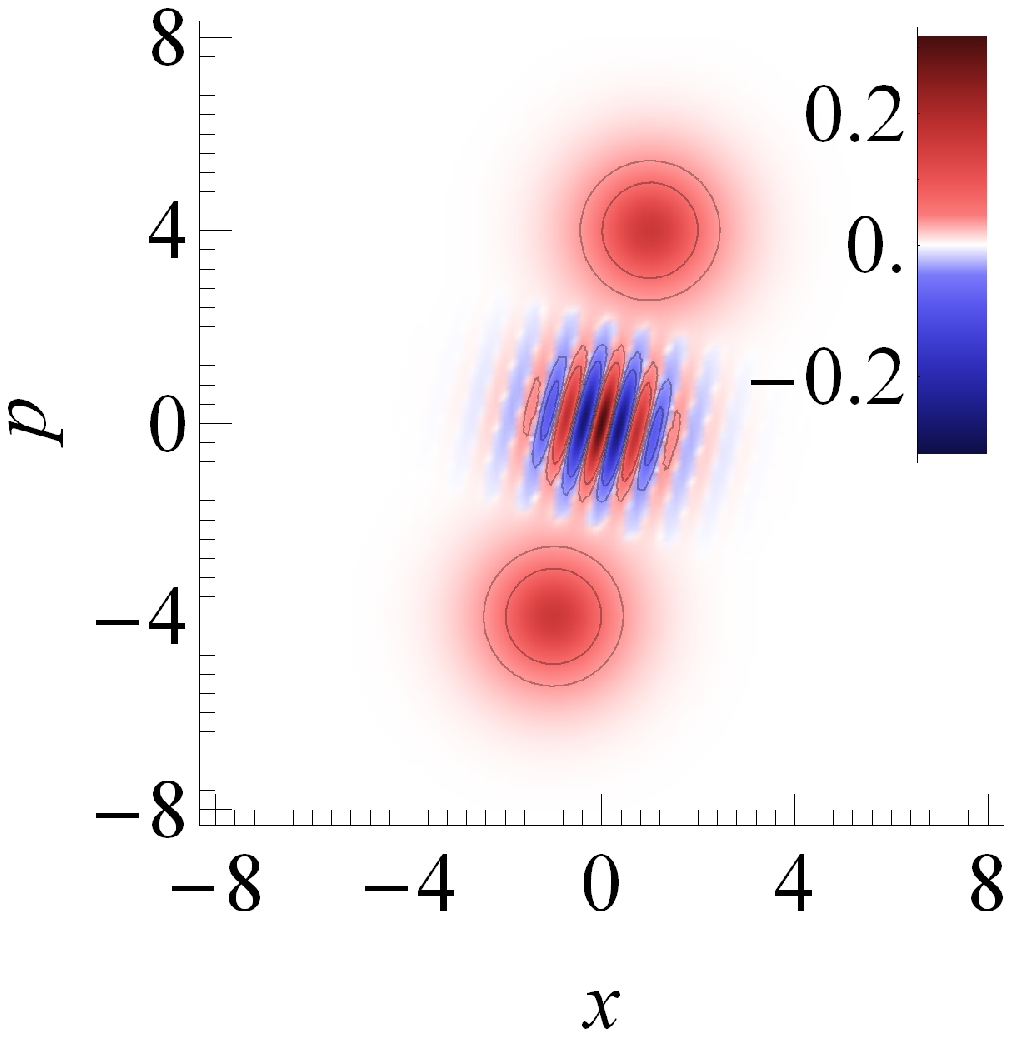}
\label{@FIG:W01a}
}
\subfloat[
]{
\includegraphics[width=0.25\textwidth]{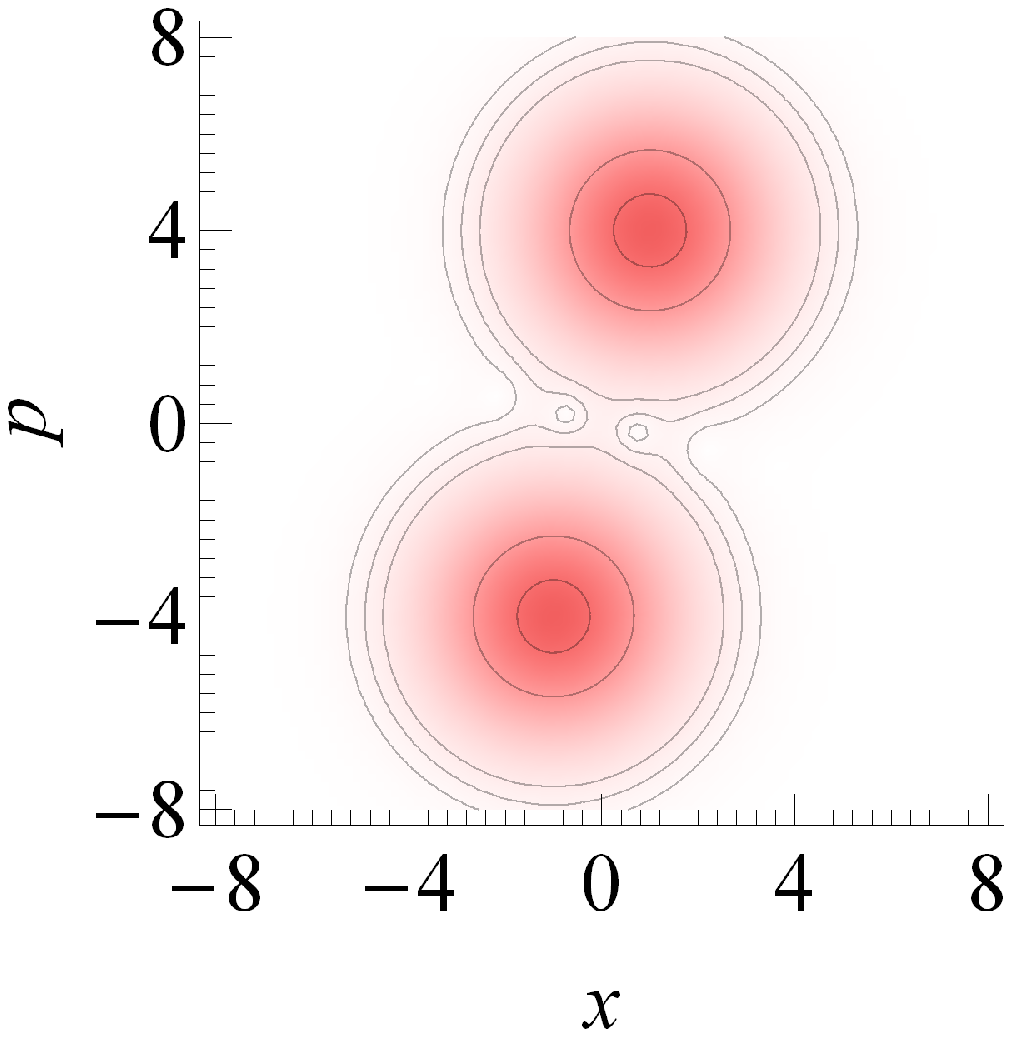}
\label{@FIG:W01b}
}
\subfloat[
]{
\includegraphics[width=0.25\textwidth]{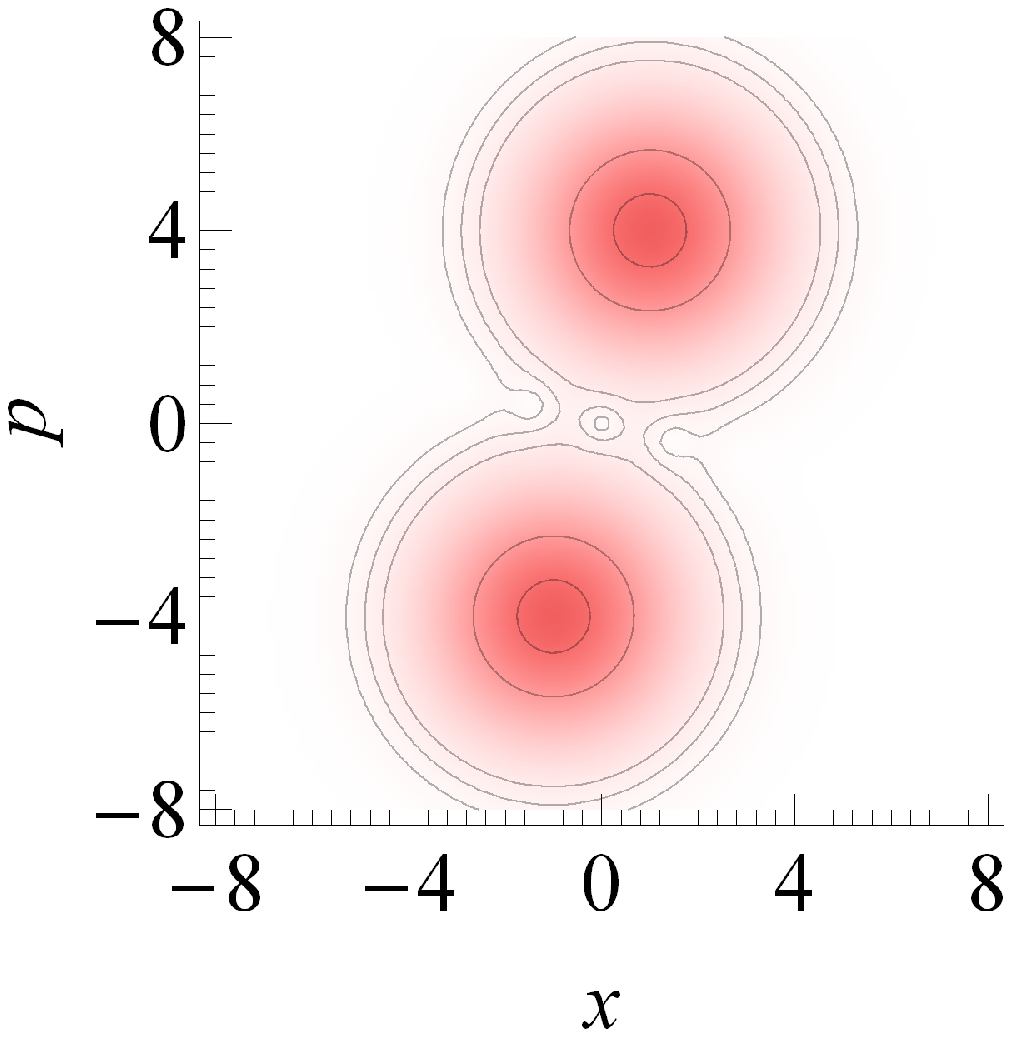}
\label{@FIG:W01c}
}
\caption{
The Wigner function defined by Eq.~\eqref{app01.-WF_of_pure_state} (panel \protect\subref{@FIG:W01a}) and the Husimi functions defined by Eq.~\protect\eqref{10.-HF_definition} 
(panels \protect\subref{@FIG:W01b} and \protect\subref{@FIG:W01c}) representing the superposition of two Gaussian states $\ket{\psi}{\propto}\left.\ket{\pc{=}4,\xc{=}1}{+}e^{i\varphi}\ket{\pc{=}{-}4,\xc{=}{-}1}\right.$. Dimensionless units where $\wc{=}\hbar{=}1$ are used. Panels \protect\subref{@FIG:W01a} and \protect\subref{@FIG:W01b} correspond to ${\varphi{=}0}$, whereas panel \protect\subref{@FIG:W01c} shows the case $\varphi{=}\pi$.\label{@FIG:W01}
}
\end{figure*}

\section{Hydrodynamic interpretation of coupled coherent states (CCS) method\label{@APP:02}}
In the CCS method, the trajectories $\{\ppc_k(t),\xxc_k(t)\}$ of centers  of basis states in the time-dependent anzatz~
\eqref{01.-cs-anzatz} are defined as the optimal trajectories minimizing the residual
\begin{gather}\label{app02.-CCS_extremal_problem}
\scpr{r_k}{r_k}\to\min|_{\ppc_k(t),\xxc_k(t)},
\end{gather}
where $\ket{r_k}{=}(\pder{}t{-}\frac{1}{i\hbar}\hat H)\scs{\ppc_k(t),\xxc_k(t)}$. In other words, the functions $\ppc_k(t)$ and $\xxc_k(t)$ provide the best fit to the solution of the Scr\"odinger equation $i\hbar\pder{}t\ket{\psi}{=}\hat H\ket{\psi}$ when using the single-Gaussian approximation $\ket{\psi}{\propto}\scs{\ppc_k(t),\xxc_k(t)}$. The solution to the problem \eqref{app02.-CCS_extremal_problem} is
\begin{subequations}\label{app02.-CCS_px_equations}
	\begin{gather}
	\tpder{}t\ppc_{k}(t){=}{-}\matel{\ppc_k(t),\xxc_k(t)}{\tpder{V(\hat{\xx})}{\hat{\xx}}}{\ppc_k(t),\xxc_k(t)},\\%
	\tpder{}t\xc_{k,n}(t){=}\matel{\ppc_k(t),\xxc_k(t)}{\tfrac{\hat p_n}{m_n}}{\ppc_k(t),\xxc_k(t)},
	\end{gather}
\end{subequations}
where the Hamiltonian of form \eqref{app01.-separable_Hamiltonian} is assumed. Straightforward computation shows that Eqs.~\eqref{app02.-CCS_px_equations} can be reduced to
\begin{gather}\label{app02.-CCS_px_equations(wigner_form)}
\tpder{}t\ppc_{k}(t){=}\midop{\WWC_{\pp}(\WF_k)},~~
\tpder{}t\xc_{k,n}(t){=}\midop{\WWC_{x_n}(\WF_k)},
\end{gather}
where $\WF_k(\pp,\xx)$ is the Wigner representation of the basis function $\scs{\ppc_k,\xxc_k}$ (i.e., it is obtained by substituting $\ket{\psi}{=}\scs{\ppc_k,\xxc_k}$ into Eq.~\eqref{app01.-WF_of_pure_state}). The Wigner flows $\WWC$ and their averaged values are defined by Eqs.~\eqref{app01.-J^W} and \eqref{app01.-<J^W>}. Thus, the CCS solutions for $\ppc_k(t)$ and $\xxc_k(t)$ allow for a simple interpretation as the trajectories guided by the averaged Wigner phase space velocities for a single basis state.

\section{Regularized Bohmian representation: Closed-form expressions for fluid velocities and fluxes\label{@APP:04}}
\subsection{Preliminaries}
Assume that a quantum state $\ket{\psi}$ is approximated using the anzatz of squeezed coherent states 
\eqref{01.-cs-anzatz}
. The respective Wigner and Husimi functions can be represented as
\begin{subequations}\label{app04.-P&Q_expansions}
\begin{align}\label{app03.-wigner_function_expansion}
\WF(\pp,\xx)&{=}\sum_{k_1{=}1}^{\basissize}\sum_{k_2{=}1}^{\basissize}\ac_{k_1}^{\ast}(t)\ac_{k_2}(t)\WF_{k_1,k_2}(\pp,\xx),\\
\HF''(\pp,\xx)&{=}\sum_{k_1{=}1}^{\basissize}\sum_{k_2{=}1}^{\basissize}\ac_{k_1}^{\ast}(t)\ac_{k_2}(t)\HF_{k_1,k_2}''(\pp,\xx),
\end{align}\label{app04.-husimi_function_expansion}
\end{subequations}
where 
\begin{widetext}
\begin{align}
\WF_{k_1,k_2}(\pp,\xx){=}&\left(\tfrac1{2\pi\hbar}\right)^{\dimensionality}\inftyints e^{-i\frac{\pp\cdot\xx'}{\hbar}}\scpr{\ppc_{k_1},\xxc_{k_1}}{\xx{-}\tfrac{\xx'}2}\scpr{\xx{+}\tfrac{\xx'}2}{\ppc_{k_2},\xxc_{k_2}}\diff^N\xx'{=}\prod_{n{=}1}^{\dimensionality}\WF_{k_1,k_2,n}(p_n,x_n),\label{app04.-P-term}\\
\WF_{k_1,k_2,n}(p_n,x_n){=}&\tfrac1{\pi  \hbar }{\exp \left({-}\tfrac{\wc_n^2}{2
		\hbar ^2}
		\left(\pc_{k_1,n}^2{+}\pc_{k_2,n}^2{+}\tfrac{2 \hbar ^2
			(x_n{-}\xcAvg{k_1}{k_2}{n})^2}{\wc_n^4}{-}2
		{\pcAvg{k_1}{k_2}{n}}^2{+}2 (p_n{-}{\pcAvg{k_1}{k_2}{n}})^2\right)\right)},
\label{app04.-P-product_term}\\
\HF_{k_1,k_2}''(\pp,\xx){=}&\KernWD_{\wwW_{\pp}'',\wwW_{\xx}''}{\cdot}\WF_{k_1,k_2}(\pp,\xx){=}\prod_{n{=}1}^{\dimensionality}\HF_{k_1,k_2,n}''(p_n,x_n),\label{app04.-Q-term}\\
\HF_{k_1,k_2,n}''(p_n,x_n){=}&\tfrac{\wc_n}{\pi\sqrt{\left(\wc_n^2{+}2
			{\varpi_{x_n}''}^2\right) \left(2 \wc_n^2
			{\varpi_{p_n}''}^2{+}\hbar ^2\right)}} \exp
	\left({-}\tfrac{\wc_n^2
		({\pcAvg{k_1}{k_2}{n}}{-}{p_n})^2}{2 \wc_n^2
		{\varpi_{p_n}''}^2{+}\hbar
		^2}{-}\tfrac{\left(\pc_{k_1,n}^2{+}\pc_{k_2,n}^2\right)
		\wc_n^2}{2 \hbar ^2}{+}\tfrac{{\pcAvg{k_1}{k_2}{n}}^2
		\wc_n^2}{\hbar
		^2}{-}\tfrac{({\xcAvg{k_1}{k_2}{n}}{-}{x_n})^2}{\wc_n^2
		{+}2 {\varpi_{x_n}''}^2}\right).
\label{app04.-Q-product_term}
\end{align}
\end{widetext}
In the above expressions we used the notations
\begin{align}
{\pcAvg{k_1}{k_2}{n}}&{=}\tfrac{1}{2} \left(\pc_{k_1,n}{+}\pc_{k_2,n}{+}\tfrac{i \hbar 
	\left(\xc_{k_1,n}{-}\xc_{k_2,n}\right)}{\wc_n^2}\right),\\
{\xcAvg{k_1}{k_2}{n}}&{=}\tfrac{1}{2} \left({-}\tfrac{i
	\left(\pc_{k_1,n}{-}\pc_{k_2,n}\right) \wc_n^2}{\hbar
}{+}\xc_{k_1,n}{+}\xc_{k_2,n}\right).
\end{align}

The following relations for $\WF_{k_1,k_2}(\pp,\xx)$ and $\HF_{k_1,k_2}''(\pp,\xx)$ will be needed in the subsequent derivations: 
\begin{align}\label{app04.-int(pP(p))dp}
\linftyint{p_n}p_n&\WF_{k_1,k_2}(\pp,\xx)\diff p_n{=}\notag\\
&{\pcAvg{k_1}{k_2}{n}}\linftyint{p_n}\WF_{k_1,k_2}(\pp,\xx)\diff p_n{+}\tfrac{\hbar^2}{2\wc_n^2}\WF_{k_1,k_2}(\pp,\xx),\\
\KernWD_{\wwW_{\pp}'',\wwW_{\xx}''}&{\cdot}\left(\linftyint{p_n}\WF_{k_1,k_2}(\pp,\xx)\diff p_n\right){=}\linftyint{p_n}\HF_{k_1,k_2}''(\pp,\xx)\diff p_n{=}\notag\\
&\JBGone{k_1}{k_2}{n}(p_n-{\pcAvg{k_1}{k_2}{n}})\HF_{k_1,k_2}''(\pp,\xx),\label{app04.-int(Q)dp}
\end{align}
where the function $\JBGone{k_1}{k_2}{n}(z)$ is defined as
\begin{gather}\label{app04.-f_n(z)}
\JBGone{k_1}{k_2}{n}(z){=}\tfrac{\sqrt{\pi }}{2 \bar{s}_n} e^{z^2 \bar{s}_n^2} \left(\erf\left(z\bar{s}_n\right){+}1\right)
\end{gather}
with $\bar{s}_n{=}\sqrt{\tfrac{\wc_n^2}{2 \wc_n^2 {\varpi_{p_n}''}^2{+}\hbar ^2}}$ ($\erf$ stands for error function).

\subsection{Regularized Bohmian momentum \texorpdfstring{$\pBohmReg_n(\ppExcl{p_n},\xx)$}{}}
Before turning to a general multidimensional case, let us demonstrate that Eq.~\eqref{12.-generalized_Bohmian_momentum} 
takes a particularly simple form for a pure state $\ket{\psi}$ in the one-dimensional case $\dimensionality{=}1$
\begin{widetext}
\begin{align}
\pBohmReg_1&(x_1){=}
\tfrac{\iiint p'
	P(p',x') \exp
	\left({-}\frac{\left(x'{-}x_1\right)^2}{2
		{\varpi_{x_1}''}^2}{-}\frac{\left(p'{-}p\right
		)^2}{2
		{\varpi_{p_1}''}^2}\right)dp'dx'dp}{\iiint
	P(p',x') \exp
	\left({-}\frac{(x'{-}x)^2}{2
		{\varpi_{x_1}''}^2}{-}\frac{(p'{-}p
		)^2}{2 {\varpi_{p_1}''}^2}\right)dp'dx'dp}{=}
\tfrac{\Re\left[\int \psi(x')^*
		({-}i \hbar)\pder{\psi
			(x')}{x'}
		\exp\left({{-}\frac{(x'{-}x_1)^2}{2
				{\varpi_{x_1}''}^2}}\right) \, dx'\right]}{\int \psi
		(x')^* \psi (x')
		\exp\left({{-}\frac{(x'{-}x_1)^2}{2
				{\varpi_{x_1}''}^2}}\right) \, dx'}{=}
\tfrac{\Re[\matel{\psi}{\exp\left({{-}\frac{(\hat x_1{-}x_1)^2}{2
			{\varpi_{x_1}''}^2}}\right)\hat p_1}{\psi}]}{\matel{\psi}{\exp\left({{-}\frac{(\hat x_1{-}x_1)^2}{2
			{\varpi_{x_1}''}^2}}\right)}{\psi}}.\label{app04.-pBohmReg-1D}
\end{align}
Once $\ket{\psi}$ is represented by the anzatz of the squeezed coherent states 
\eqref{01.-cs-anzatz}
, Eq.~\eqref{app04.-pBohmReg-1D} can be evaluated analytically using the well-known technique described, e.g., in Ref.~\cite{2008-Shalashilin}.

The above derivation can be straightforwardly generalized to a multidimensional case
\begin{gather}\label{app04.-pBohmReg}
\pBohmReg_n(\ppExcl{p_n},\xx){=}\frac{\sum_{k_1,k_2{=}1}^{\basissize}
	\ac_{k_1}^{\ast}(t)\ac_{k_2}(t)
\matel{\pc_{k_1,n},\xc_{k_1,n}}{\{\hat p_n,e^{{-}\frac{(\hat x_n{-}x_n)^2}{2
			{\varpi_{x_1}''}^2}}\}_+}{\pc_{k_2,n},\xc_{k_2,n}}
\prod_{n'{\ne}n}\HF_{k_1,k_2,n'}''(p_{n'},x_{n'})}{2\sum_{k_1,k_2{=}1}^{\basissize}\ac_{k_1}^{\ast}(t)\ac_{k_2}(t)
\matel{\pc_{k_1,n},\xc_{k_1,n}}{e^{{-}\frac{(\hat x_n{-}x_n)^2}{2
			{\varpi_{x_1}''}^2}}}{\pc_{k_2,n},\xc_{k_2,n}}\prod_{n'{\ne}n}\HF_{k_1,k_2,n'}''(p_{n'},x_{n'})}.
\end{gather}
\end{widetext}

\subsection{Husimi fluxes}
The position flux components ${\GGBC_{x_n}}(\pp,\xx)$ can be readily evaluated by substituting Eqs.~\eqref{app04.-husimi_function_expansion}, \eqref{app04.-Q-term} and \eqref{app04.-Q-product_term} into Eq.~\eqref{12.-generalized_bohmian_position_flow}.

The computation of the momentum flux components
\begin{gather}\label{app04.-generalized_bohmian_momentum_flow}
{\GBC_{p_n}}(\pp,\xx){=}{\HC_{p_n}}(\pp,\xx){+}\delta\GBC_{p_n}(\pp,\xx),
\end{gather}
is more involving. (Recall that the ``default'' component $\HHC_{p_n}(\pp,\xx)$ of the momentum Husimi flow is defined by Eqs.~\eqref{app09.-J^H} and \eqref{app01.-J^W}). We are going to show how one can evaluate both the terms in Eq.~\eqref{app04.-generalized_bohmian_momentum_flow}. 

The gauge term $\delta{\GBC_{p_n}}(\pp,\xx)$ is given by the first of Eqs.~\eqref{11.-gauge_flows_H}. 
In order to evaluate the right hand side of Eqs.~\eqref{11.-gauge_flows_H}, one needs to know the gauge potential defined in 
Eq.~\eqref{12+.-gen_bohmian_gauge_potential}. One can proceed by expanding ${\GPGB_n}(\pp,\xx)$ analogously to Eqs.~\eqref{app04.-P&Q_expansions}
\begin{gather}
{\GPGB_n}(\pp,\xx){=}\sum_{k_1{=}1}^{\basissize}\sum_{k_2{=}1}^{\basissize}\ac_{k_1}^{\ast}\ac_{k_2}{\GPGB_{n,k_1,k_2}}(\pp,\xx),
\end{gather}
where
\begin{align}
{\GPGB_{n,k_1,k_2}}&(\pp,\xx){=}
\tfrac1{m_n}\big(\pBohmReg_n\linftyint{p_n}\HF_{k_1,k_2}''(\pp,\xx)\diff p_n{-}\notag\\
&\linftyint{p_n}\KernWD_{\wwW_{\pp}'',\wwW_{\xx}''}{\cdot}p_n\WF_{k_1,k_2}(\pp,\xx)\diff p\big).
\end{align}
The terms ${\GPGB_{n,k_1,k_2}}(\pp,\xx)$ can be evaluated via Eqs.~\eqref{app04.-int(pP(p))dp} and \eqref{app04.-int(Q)dp}
\begin{align}\label{app04.-A^H_term}
{\GPGB_{n,k_1,k_2}}(\pp,\xx)&{=}\tfrac1{m_n}\big((\pBohmReg_n{-}{\pcAvg{k_1}{k_2}{n}})\JBGone{k_1}{k_2}{n}(p_n{-}{\pcAvg{k_1}{k_2}{n}})+\notag\\&\tfrac{\hbar^2}{2\wc_n^2}\big)\HF_{k_1,k_2}''(\pp,\xx).
\end{align}
Substitution of Eq.~\eqref{app04.-A^H_term} into the first of Eqs.~\eqref{11.-gauge_flows_H} 
gives
\begin{gather}\label{app03.-position_gauge_flux_expansion}
{\delta\GBC_{p_n}}(\pp,\xx){=}\sum_{k_1{=}1}^{\basissize}\sum_{k_2{=}1}^{\basissize}\ac_{k_1}^{\ast}\ac_{k_2}{{\delta\GBC_{p_n,k_1,k_2}}}(\pp,\xx),
\end{gather}
where
\begin{align}
{\delta\GBC_{p_n,k_1,k_2}}&(\pp,\xx){=}\notag\\
&\tfrac1{m_n}\bigg(\tpder{\pBohmReg_n}{x_n}\JBGone{k_1}{k_2}{n}(p_n{-}{\pcAvg{k_1}{k_2}{n}})\HF_{k_1,k_2}''(\pp,\xx){+}\notag\\
&\big((\pBohmReg_n{-}{\pcAvg{k_1}{k_2}{n}})\JBGone{k_1}{k_2}{n}(p_n{-}{\pcAvg{k_1}{k_2}{n}})+\notag\\
&\tfrac{\hbar^2}{2\wc_n^2}\big)\tpder{\HF_{k_1,k_2}''(\pp,\xx)}{x_n}\bigg).
\end{align}

Let us now turn to computing the first term ${\HHC_{p_n}}(\pp,\xx)$ in Eq.~\eqref{app04.-generalized_bohmian_momentum_flow}. We will restrict ourselves to the case of the Hamiltonians specified in 
Eq.~\eqref{02'.-Hamiltonian} with the separable potential 
\begin{gather}\label{app04.-separable_potential}
V(\xx){=}\sum_nV_n(x_n),
\end{gather}
as this case allows for the exact analytical treatment.

We again start by expanding the momentum flow $\HHC_{\pp}$ defined in Eqs.~\eqref{app09.-J^H} and \eqref{app01.-J^W} similarly to Eq.~\eqref{app03.-position_gauge_flux_expansion}
\begin{gather}\label{app03.-position_flux_expansion}
{\HC_{p_n}}(\pp,\xx){=}\sum_{k_1{=}1}^{\basissize}\sum_{k_2{=}1}^{\basissize}\ac_{k_1}^{\ast}\ac_{k_2}\HC_{p_n,k_1,k_2}(\pp,\xx),
\end{gather}
where
\begin{gather}
\HC_{p_n,k_1,k_2}(\pp,\xx){=}\HC_{p_n,k_1,k_2,n}(p_n,x_n)\prod_{n_1{\ne}n}\HF_{k_1,k_2,n_1}''(p_{n_1},x_{n_1}),
\end{gather}
and
\begin{align}\label{app04.-J^H_{p_n,xi}}
\HC_{p_n,k_1,k_2,n}&(p_n,x_n){=}\KernWD_{\wW_{p_n}'',\wW_{x_n}''}{\cdot}\WC_{p_n,k_1,k_2,n}(p_n,x_n),
\\
\WC_{p_n,k_1,k_2,n}&(p_n,x_n){=}\notag\\
{-}\WF_{n,k_1,k_2}&(p_n,x_n)\sinc\big(\tfrac{\hbar}2(-\pderl{p_n}{\cdot}\pderr{x_n})\big)\tpder{V_n(x_n)}{x_n}.\label{app04.-J^W_{p_n,xi}}
\end{align}
To simplify notations, hereafter we will use the composite index $\xi{=}\{k_1,k_2,n\}$, i.e., $\HC_{p_n,\xi}{\equiv}\HC_{p_n,k_1,k_2,n}$ etc.
Following the discussion in Sec.~\ref{@APP:01.-key_concepts}, we can reexpress Eq.~\eqref{app04.-J^W_{p_n,xi}} in terms of the Blokhintsev function $\tilde{\WF_{\xi}}(\lambda,x_n){=}\IFT{p_n}{\lambda}[\WF_{\xi}(p_n,x_n)]$ by applying the Fourier transform \eqref{app01.-Fourier_transform_operator}
\begin{gather}
\WC_{p_n,\xi}(p_n,x_n){=}{-}\FT{\lambda}{p_n}{\cdot}\bigg(\tilde{\WF}_{\xi}(\lambda,x_n)\sinc\big(\tfrac{i}2(\lambda{\cdot}\tpder{}{x_n})\big)\tpder{V_n(x_n)}{x_n}\bigg){=}\notag\\
2i\FT{\lambda}{p_n}{\cdot}\bigg(\frac{1}{\lambda}\tilde{\WF}_{\xi}(\lambda,x_n)\sin\big(\tfrac{i}2(\lambda{\cdot}\tpder{}{x_n})\big)V_n(x_n)\bigg){=}\notag\\
\FT{\lambda}{p_n}{\cdot}\bigg(\tfrac{1}{\lambda}\tilde{\WF}_{\xi}(\lambda,x_n)(V_n(x_n{+}\tfrac{\lambda}2)-V_n(x_n{-}\tfrac{\lambda}2))\bigg).
\end{gather}
The respective expression \eqref{app04.-J^H_{p_n,xi}} for $\HC_{p_n,\xi}(p_n,x_n)$ can be simplified with the help of convolution theorem
\begin{align}\label{HC_{p_n,k_1,k_2,n}-prettifying}
\HC_{p_n,\xi}&(p_n,x_n){=}\inftyint\FT{\lambda}{p_n}{\cdot}\bigg(\etilde{\KernW}_{\wW_{p_n}'',\wW_{x_n}''}(\lambda,x_n{-}x')\notag\\
&\tfrac{1}{\lambda}\tilde{\WF}_{\xi}(\lambda,x')(V_n(x'{+}\tfrac{\lambda}2)-V_n(x'{-}\tfrac{\lambda}2))\bigg)\diff x',
\end{align}
where \begin{align}
\etilde{\KernW}_{\wW_{p_n}'',\wW_{x_n}''}(\lambda,\delta x)&{=}\IFT{p'}{\lambda}[{\KernW}_{\wW_{p_n}'',\wW_{x_n}''}(p_n{-}p',\delta x)]{=}\notag\\
\tfrac{1}{2 \pi  \sqrt{\hbar } {\varpi_{x_n}''}}&\exp \left({-}\tfrac{\text{$\delta $x}^2}{2
		{\varpi_{x_n}''}^2}{-}\tfrac{\lambda  \left(\lambda 
		{\varpi_{p_n}''}^2{+}2 i p_n \hbar \right)}{2 \hbar
		^2}\right)
\end{align}
is the partially Fourier-transformed kernel $\KernW_{\wW_{p_n}'',\wW_{x_n}''}$ (see Eq.~\eqref{10.-Husimi_kernel}
).

Recall that the squeezed coherent states have the form
\begin{align}\label{app04.-CS_definition'}
\scpr{\xx}{\ppc_k,\xxc_k}&{=}\prod_{n=1}^{\dimensionality}\chi_{k,n}(x_n),\notag\\ \chi_{k,n}(x_n){=}&\prod_{n{=}1}^{N}\tfrac1{\sqrt{\pi}\wc_n^{\frac14}}e^{-\frac{(x_n{-}\xc_{k,n})^2}{2\wc_n^2}{+}\frac i{\hbar}\pc_{k,n}(x_n{-}\xc_{k,n})}.
\end{align}
Using this definition and with the help of the substitutions $y{=}x{-}\frac{\lambda}2$, $z{=}x{+}\frac{\lambda}2$, Eq.~\eqref{HC_{p_n,k_1,k_2,n}-prettifying} can be reshaped as
\begin{gather}
\HC_{p_n,\xi}(p_n,x_n){=}\tfrac{I^{(0)}_{k_1,k_2,n}(p_n,x_n){+}(I^{(0)}_{k_2,k_1,n}(p_n,x_n))^{*}}{\sqrt{2\pi\hbar}},
\end{gather}
where 
\begin{gather}\label{app04-integral_I^(0)}
I^{(0)}_{\xi}(p_n,x_n){=}\inftyint\chi_{k_1,n}(y)V_n(y)I^{(1)}_{k_2,n}(y,p_n,x_n)\diff y,
\end{gather}
and
\begin{gather}
I^{(1)}_{k_2,n}(y,p_n,x_n){=}\inftyint \tfrac{\etilde{\KernW}_{\wW_{p_n}'',\wW_{x_n}''}(z{-}y,\frac{z+y}2{-}x_n)}{z{-}y}\diff z.
\end{gather}
The last integral is of a general form
\begin{gather}
C_0\inftyint\tfrac{e^{-A_0 \lambda^2+B_0 \lambda}}{\lambda}\diff \lambda {=}\pi C_0\erfi\left(\tfrac {B_0}{\sqrt 2 A_0}\right),
\end{gather}
where, in our case,
\begin{align*}
C_0{=}&\frac{C_1}{\pi}e^{-A_1 y^2+B_1y},~~
A_1{=}\tfrac{1}{2}
\left(\tfrac{1}{\wc_n^2}{+}\tfrac{1}{{\varpi_{x_
			n}''}^2}\right),\\
B_1{=}&\tfrac{i \pc_{k_2,n}}{\hbar
}{+}\tfrac{\xc_{k_2,n}}{\wc_n^2}{+}\tfrac{x_n}{{
		\varpi_{x_n}''}^2},\\
C_1{=}&\tfrac{\exp \left({-}\tfrac{x_n^2}{2
		{\varpi_{x_n}''}^2}{-}\tfrac{1}{2} \xc_{k_2,n}
	\left(\tfrac{\xc_{k_2,n}}{\wc_n^2}{+}\tfrac{2 i
		\pc_{k_2,n}}{\hbar }\right)\right)}{2 \pi
	^{1/4} {\varpi_{x_n}''} \sqrt{\hbar  \wc_n}},\\
A_0{=}&\tfrac{1}{8} \left(\tfrac{4}{\wc_n^2}{+}\tfrac{4
	{\varpi_{p_n}''}^2}{\hbar
	^2}{+}\tfrac{1}{{\varpi_{x_n}''}^2}\right),~~
B_0{=}B_{0,1}y+B_{0,2},\\
B_{0,1}{=}&{-}\tfrac{1}{\wc_n^2}{-}\tfrac{1}{2
	{\varpi_{x_n}''}^2},~~
B_{0,2}{=}{-}\tfrac{i \left(p_n{-}\pc_{k_2,n}\right)}{\hbar
}{+}\tfrac{\xc_{k_2,n}}{\wc_n^2}{+}\tfrac{x_n}{2
	{\varpi_{x_n}''}^2}.
\end{align*}
Thus,
\begin{gather}
I^{(1)}_{k_2,n}(y,p_n,x_n){=}C_1e^{-A_1 y^2+B_1y}\erfi(B_2 y+C_2),
\end{gather}
where
$B_2{=}\frac{B_{1,1}}{\sqrt 2A_0}$, $C_2{=}\frac{B_{1,2}}{\sqrt 2A_0}$ and $\erfi(z){=}\tfrac{\erf(iz)}i$.

Assume that the potential $V_n(y)$ can be cast as a sum
\begin{gather}
V_n(y){=}\sum_k v_{n,k},
\end{gather}
where each term $v_{n,k}$ has the following generic form
\begin{gather}\label{app04-potential_term}
v_{n,k}(y){=}c_{r(k)} y^{r(k)} e^{v_2(k) y^2{+}v_1(k) y}
\end{gather}
with non-negative integers $r$ and some constants $c_r$, $v_1$ and $v_2$. In this case, the integral \eqref{app04-integral_I^(0)} is expandable as
\begin{gather}
I^{(0)}_{\xi}(p_n,x_n){=}\sum_k c_{r(k)}I^{(0)}_{\xi,k}(p_n,x_n),
\end{gather}
where each $I^{(0)}_{\xi,k}(p_n,x_n)$ is of the form
\begin{gather}
I^{(0)}_{\xi,k}(p_n,x_n){=}\inftyint\chi_{k_1,n}(y)y^{r} e^{v_2 y^2{+}v_1 y}I^{(1)}_{k_2,n}(y,p_n,x_n)\diff y{=}\notag\\C_4\inftyint y^r e^{-\alpha y^2{+}\beta y}\erfi(B_2y+C_2)\diff y{=}\notag\\
-iC_4\sqrt{2 \pi } (2 \alpha )^{{-}\frac{n{+}1}{2}
	} I_r^{(2)}\left(i\tfrac{B_2}{\sqrt{2 \alpha
}},iC_2,\tfrac{\beta }{\sqrt{2 \alpha }}\right).
\end{gather}
Here
\begin{align*}
\alpha{=}&\tfrac{1}{2 \wc_n^2}{+}A_1{+}v_2,~~\beta{=}{-}\tfrac{i \pc_{k_1,n}}{\hbar
}{+}\tfrac{\xc_{k_1,n}}{\wc_n^2}{+}B_1{+}v_1,\\
C_4{=}&\tfrac{C_1 e^{\xc_{k_1,n}
		\left({-}\frac{\xc_{k_1,n}}{2
			\wc_n^2}{+}\frac{i \pc_{k_1,n}}{\hbar
		}\right)}}{\sqrt[4]{\pi } \sqrt{\wc_n}}
\end{align*}
and
\begin{gather}
	I_r^{(2)}(a,b,\mu){=}\tfrac{1}{\sqrt{2 \pi }}\inftyint x^r e^{{-}\frac{x^2}{2}{+}\mu x} \erf(ax{+}b)\diff x{=}\notag\\
	P_r(\mu)
	e^{\frac{\mu ^2}{2}} \text{erf}\left(\tfrac{a \mu
		{+}b}{\sqrt{2 a^2{+}1}}\right)
	{+} Q_r(\mu)
	e^{\frac{\mu ^2}{2}{-}\frac{(a \mu {+}b)^2}{2
			a^2{+}1}}
	,
\end{gather}
where the polynomials $P_r(\mu)$ and $Q_r(\mu)$ are defined recursively as
\begin{align}
P_0(\mu){=}&1,~~Q_0(\mu){=}0,\notag\\
P_{r+1}(\mu){=}&\tpder{P_r(\mu )}{\mu}{+}\mu P_r(\mu),\notag\\
Q_{r+1}(\mu){=}&\tfrac{(\mu {-}2 a b) Q_r(\mu )}{2
	a^2{+}1}{+}\tfrac{2 a P_r(\mu )}{\sqrt{\pi(2 
		a^2{+}1)}}{+}\tpder{Q_r(\mu )}{\mu}.\notag
\end{align}
\subsection{Remark on multidimensional computations}
The outlined procedures for computing the components of $\ppBohmReg$, ${\HHC}_{\xx}''$ and ${\delta\HHC}_{\pp}''$ are applicable to both one- and multidimensional systems. The numerical complexities scale linearly with the number of dimensions. However, the described algorithm for computing the ``default'' Husimi momentum flow (given by the first of Eqs.~\eqref{app09.-J^H}, where $\WWC_{\pp}$ is defined by Eq.~\eqref{app01.-J^W}) is effective only for one-dimensional systems, and its direct extension to multidimensional systems (apart from the special cases of harmonic and separable potentials \eqref{app04.-separable_potential}) is prohibitively computationally expensive. Here we demonstrate the feasibility of the approximation
\begin{align}
\HHC_{\pp}(\pp,\xx)&{\simeq}{-}\Re\bigg[\KernWD_{\wwW_{\pp}'',\wwW_{\xx}''}\bigg(\tfrac1{(2\pi\hbar)^{\dimensionality}}\inftyints\diff\llambda^{\dimensionality}\times\notag\\
&\scpr{\psi}{\xx{-}\tfrac{\llambda}2}\matel{{\xx{+}\tfrac{\llambda}2}}{\tpder{\hat H}{\hat x}}{\psi}e^{-i\frac{\pp\cdot\llambda}{\hbar}}\bigg)\bigg],\label{app07.-J^H_p-thermal_approximation}
\end{align}
applicable to the case of a pure quantum state $\ket{\psi}$. Note that Eq.~\eqref{app07.-J^H_p-thermal_approximation} reduces to 
\begin{gather}\label{app07.-J^H_p-non-thermal_approximation}
\HHC_{\pp}(\pp,\xx){\simeq}{-}\Re[\matel{\pp,\xx}{\tpder{\hat H}{\hat x}}{\psi}\scpr{\psi}{\pp,\xx}],
\end{gather}
in the case of the conventional ``non-thermal'' Husimi representation defined in 
Eq.~\eqref{10.-HF_definition}. The procedures for computing right hand side of Eq.~\eqref{app07.-J^H_p-non-thermal_approximation} in a multidimensional case are well-developed \cite{2008-Shalashilin}. Similarly, Eq.~\eqref{app07.-J^H_p-thermal_approximation} can be analytically evaluated for generic operators $\hat O{=}{-}\tpder{\hat H}{\hat x}$ of the form
\begin{gather}
\hat O{=}\sum_{k,s,r}a_{k,s,r}\prod_{n{=}1}^{\dimensionality}\hat p^s_n\hat x^r_n e^{c_{2,k,n} \hat x_n^2{+}c_{1,k,n} \hat x_n},
\end{gather}
where $a_{k,s,r}, c_{1,k,n}{\in}\complexes$ and $c_{2,k,n}{\in}\reals$ are some coefficients. Using the expansion similar to \eqref{app03.-position_flux_expansion}, the only non-trivial part of evaluating Eq.~\eqref{app07.-J^H_p-thermal_approximation} is the calculation of an one-dimensional integrals of the form
\begin{align}
\KernWD&_{\wW_{p}'',\wW_{x}''}\big(\tfrac1{2\pi\hbar}\inftyint\diff \lambda e^{-i\frac{p\lambda}{\hbar}}\times\notag\\
&\scpr{\pc_1,\xc_1}{x{-}\tfrac{\lambda}2}\matel{x{+}\tfrac{\lambda}2}{\hat p^s\hat x^r e^{c_2 \hat x^2{+}c_1 \hat x}}{\pc_2,\xc_2}\big).\label{app07.-int1}
\end{align}
In the following, it will be helpful to explicitly introduce the width parameter into the notation \eqref{app04.-CS_definition'} for the squeezed coherent states
\begin{gather}\label{app07.-exp*scs_relation}
\ket{\ppc,\xxc}{\equiv}\ket{\ppc,\xxc,\wwc}
\end{gather}
One can check that the following relation holds:
\begin{gather}
e^{c_2 \hat x^2{+}c_1 \hat x}\ket{\pc,\xc,\wc}{\equiv}c'\ket{\pc',\xc',\wc'},
\end{gather}
where
\begin{gather*}
	c'{=}\tfrac{\exp\left({{-}\frac{\frac{\wc^2 \left(2 i \pc \left(2 c_2 \xc{+}\Re[c_1]\right){+}\Re[c_1] \hbar  (c_1{+}i \Im[c_1])\right)}{\hbar }{+}2 \xc \left(c_2 \xc{+}c_1\right)}{2 \left(2 c_2 \wc^2{-}1\right)}}\right)}{\sqrt[4]{1{-}2 c_2 \wc^2}},\\
	\pc'{=}\pc{+}\Im[c_1] \hbar,~~\xc'{=}\tfrac{\Re[c_1]
	\wc^2{+}\xc}{1{-}2 c_2 \wc^2},~~
	\wc'{=}\tfrac{\wc}{\sqrt{{1}{-}{2c_2\wc^2}}}.
\end{gather*}
Using the relation \eqref{app07.-exp*scs_relation}, we can reduce the integrals in Eq.~\eqref{app07.-int1} to the evaluation of the following generic term:
\begin{align}
\KernWD&_{\wW_{p}'',\wW_{x}''}\big(\tfrac1{2\pi\hbar}\inftyint\diff \lambda e^{-i\frac{p\lambda}{\hbar}}\times\notag\\
&\scpr{\pc_1,\xc_1,\wc_1}{x{-}\tfrac{\lambda}2}\matel{x{+}\tfrac{\lambda}2}{\hat p^s\hat x^r}{\pc_2,\xc_2,\wc_2}\big).\label{app07.-int2}
\end{align}
Further simplifications can be done by using the equality
\begin{gather}
\hat p\ket{\pc,\xc,\wc}{=}\left(i\tfrac{\hbar}{\wc^2}(x{-}\xc){+}{ \pc}\right)\ket{\pc,\xc,\wc}
\end{gather}
and the canonical commutation relations between position and momentum operators, which allow to rewrite Eq.~\eqref{app07.-int2} as the sum of terms like
\begin{align}\label{app07.-int3}
I^{(4)}&{=}\KernWD_{\wW_{p}'',\wW_{x}''}\big(\tfrac1{2\pi\hbar}\inftyint\diff \lambda e^{-i\frac{p\lambda}{\hbar}}\times\notag\\
&\scpr{\pc_1,\xc_1,\wc_1}{x{-}\tfrac{\lambda}2}\matel{x{+}\tfrac{\lambda}2}{\hat x^r}{\pc_2,\xc_2,\wc_2}\big).
\end{align}
The latter allow for analytical integration
\begin{align}
I^{(4)}{=}&\tfrac{1}{2 \pi  \sqrt{F\wc_1 \wc_2} {\varpi_{p_n}''}
	{\varpi_{x_n}''}}C^rH_r\left(\tfrac{Z(p,x)}{2C}\right)\times\notag\\
&e^{A (X{-}\mathit{x})^2{+}B(P{-}\mathit{p})^2{+}D{+}C(P{-}p) (X{-}x)},
\end{align}
where $H_r$ denote Hermite polynomials, 
\begin{gather}
Z(p,x){=}\tfrac{\left(\wc_1^2 (2 A (X{-}x){+}C (P{-}p)){+}i \hbar  (2 B (p{-}P){+}C (x{-}X))\right)}{\wc_1^2/\wc_2^2{+}1}{+}X
\end{gather}
and the coefficients are
\begin{align*}
E{=}&\tfrac{1}{2} \left(\tfrac{1}{\wc_1^2}{+}\tfrac{1}{\wc_2^2}{+}\tfrac{1}{{\wW_{x}''}^2}\right),~~
F{=}\tfrac{\hbar ^2 \left(\wc_1^2{+}\wc_2^2{+}4 {\wW_{x}''}^2\right)}{8 \wc_1^2 \wc_2^2 {\wW_{p}''}^2 {\wW_{x}''}^2}{+}E,\\
A{=}&{-}\tfrac{\left(\wc_1^2{+}\wc_2^2\right) {\wW_{p}''}^2{+}\hbar ^2}{4 F \wc_1^2 \wc_2^2 {\wW_{p}''}^2 {\wW_{x}''}^2},~~
B{=}{-}\tfrac{E}{2 F {\wW_{p}''}^2},\\
C{=}&{-}\tfrac{i \hbar  \left(\frac{1}{\wc_1^2}{-}\frac{1}{\wc_2^2}\right)}{4 F {\varpi_{p_n}''}^2 {\varpi_{x_n}''}^2},~~
X{=}\tfrac{\pc_{1} \wc_1^2{+}\pc_{2} \wc_2^2{+}i \hbar 
	\left(\xc_{1}{-}\xc_{2}\right)}{\wc_1^2{+}\wc_2^2},\\
P{=}&\tfrac{\hbar  \xc_{2} \wc_1^2{+}\wc_2^2 \left(\hbar  \xc_{1}{-}i
	\left(\pc_{1}{-}\pc_{2}\right) \wc_1^2\right)}{\hbar 
	\left(\wc_1^2{+}\wc_2^2\right)},\\
D{=}&\tfrac{i}{2\hbar }\left(\pc_{1} \xc_{1}{-}\pc_{2} \xc_{2}{+}P 
	\left(\xc_{1}{-}\xc_{2}\right){-}X \left(\pc_{1}{-}\pc_{2}\right)\right).
\end{align*}

Note that Eq.~\eqref{app07.-J^H_p-non-thermal_approximation} is exact for the case of a harmonic system. In the case of non-harmonic potentials, it can be regarded as a local harmonic approximation of the potential energy surface near the phase space point $\{\pp,\xx\}$. However, Fig.~\ref{@FIG:W04} demonstrates that this approximation remains viable even for manifestly anharmonic, rapidly varying potentials. We consider again the model from Fig.~\ref{@FIG:W02}
. We reproduce the solution computed using the exact regularized Bohmian gauge (blue curves) and repeat the same calculation using the approximate expression \eqref{app07.-J^H_p-non-thermal_approximation} (red curves). One can see that the accuracies of both the methods are nearly identical up to $t{\simeq}7.5$; after that the exact and approximate solutions rapidly diverge due to an insufficient basis size.
\begin{figure}[htp]
\centering
\includegraphics[width=0.5\textwidth]{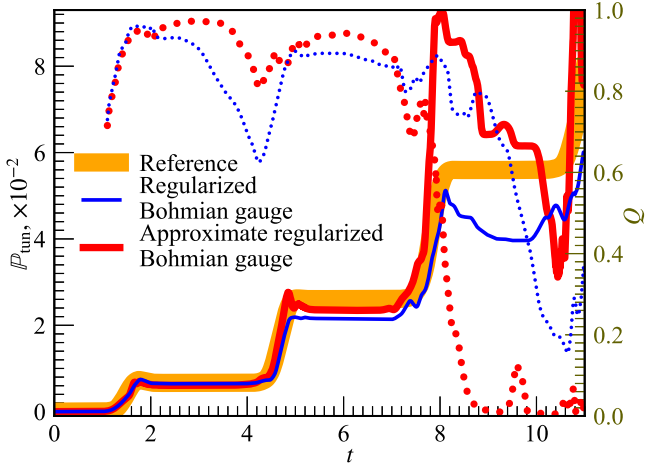}
\caption{The performance of the regularized-Bohmian-gauge-based numerical solution to the Schr\"odinger equation, where the approximation \eqref{app07.-J^H_p-non-thermal_approximation} to Husimi fluxes \eqref{app09.-J^H} is employed (red curves). The same test problem as in Fig.~\ref{@FIG:W02} 
is considered. For comparison, the respective calculations from Fig.~\ref{@FIG:W02} are reproduced (blue curves), where no approximations to the right hand side of Eqs.~\eqref{app09.-J^H} are made. The exact solution is also shown in orange as a reference. The remaining notations, parameters, and computational basis are detailed in the caption to Fig.~\ref{@FIG:W02} 
and in \APPREF{@APP:06}.\label{@FIG:W04}}
\end{figure}

\section{Numerical example: Computational details\label{@APP:06}}
\begin{table}[tbh]
	\begin{tabular}{rS[table-format=2.8]S[table-format=2.7]rS[table-format=3.7]}
		\hline
		\multicolumn{1}{r}{$k$} &  \multicolumn{1}{c}{$\ac_k$}  &  \multicolumn{1}{c}{$\mathrm{arg}(\ac_k)$} & \multicolumn{1}{c}{$\pc_k$} & \multicolumn{1}{c}{$\xc_k$}\\
		\hline
		0 & 0.424935  & -1.52358   &  -2 & -2.05    \\
		1 & 0.160277  & -0.27583   &  -2 & -1.69    \\
		2 & 0.16441   & -0.232555  &   0 & -1.93    \\
		3 & 0.326508  &  0.488428  &   0 & -1.57    \\
		4 & 0.227026  &  2.76771   &   2 & -1.81    \\
		5 & 0.0897529 &  1.74287   &   2 & -1.45    \\
		6 & 0.363191  &  2.77319   &  -2 & -2.35    \\
		7 & 0.426131  & -0.369591  &   0 & -2.25    \\
		8 & 0.268861  &  0.0324824 &   2 & -2.15    \\
		9 & 0.151384  & -2.35028   &  -2 & -6.08333 \\
		10 & 0.201003  &  2.18152   &  -2 & -5.08333 \\
		11 & 1.09911   & -0.959581  &  -2 & -4.08333 \\
		12 & 0.159998  &  2.25803   &  -2 & -3.08333 \\
		13 & 1.09394   &  0.37624   &   0 & -5.75    \\
		14 & 3.92821   & -0.214463  &   0 & -4.75    \\
		15 & 6.64557   &  0.282041  &   0 & -3.75    \\
		16 & 1.75203   & -0.368937  &   0 & -2.75    \\
		17 & 0.391672  & -0.246873  &   2 & -5.41667 \\
		18 & 0.0519809 &  1.89087   &   2 & -4.41667 \\
		19 & 0.925349  & -1.1339    &   2 & -3.41667 \\
		20 & 0.593327  &  2.0164    &   2 & -2.41667 \\
		\hline
	\end{tabular}
	\caption{The initial expansion coefficients for representing the wavefunction $\psi(t{=}0)$ in the anzatz of 21 squeezed coherent states.\label{@TBL:num_scs_basis}}
\end{table}
Here we detail the procedure for solving the time-dependent Schrodinger equation
\begin{gather}\label{app06.-Schrodinger_equation}
\ket{r}{=}\tpder{}t\ket{\psi}{-}\tfrac1{i\hbar}\hat H\ket{\psi}{=}0.
\end{gather}
within the anzatz of time-dependent squeezed coherent states 
\eqref{01.-cs-anzatz}
, which is used to obtain the results presented in Fig.~\ref{@FIG:W02} 
and Fig.~\ref{@FIG:W04}. The problem of propagating $\ket{\psi(t)}$ reduces to determining $\zzcCount$ time-dependent parameters in the anzatz \eqref{01.-cs-anzatz}, such as amplitudes $\ac_k(t)$, for which we do not provide analytical expressions. Let us collectively denote these parameters as $\zzc$. Then, the time derivative $\pder{}t\ket{\psi}$ in \eqref{app06.-Schrodinger_equation} can be written as
\begin{gather}\label{app06.-d_psi/d_t}
\tpder{}t\ket{\psi}{=}\sum_{l{=}1}^{\zzcCount}\zct_k\ket{\psi_{\zc_k}}{+}\ket{\psi_{t}},
\end{gather}
where the shorthand notations $\zct_k{=}\tder{\zc_k}t$,  $\ket{\psi_{\zc_k}}{=}\tpder{}{\zc_k}\ket{\psi}$ are used, and $\ket{\psi_{t}}{=}\tpder{}{t}\ket{\psi}$ denotes the partial time derivative over all parameters other than $\zzc(t)$. Similarly to the CCS approach \cite{2008-Shalashilin}, we determine $\zct_k$ as the coefficients minimizing the regularized least-square approximation error
\begin{gather}\label{app06.-Schrodinger_extermal_problem}
\scpr{r}{r}+\epsilon\zzct^{\intercal}D\zzct{\to}\left.\min\right|_{\zzct},
\end{gather}
where $D$ is the real diagonal $\zzcCount{\times}\zzcCount$ regularization matrix (chosen here to be equal to the identity matrix) and $\epsilon{>}0$ is an empirically chosen small parameter.
Assuming that all of the parameters $\zzc$ are real%
\footnote{Complex parameters can be included, e.g, by representing them in the polar form and treating the amplitude and phase as two new real parameters.}
and using the notations~\eqref{app06.-d_psi/d_t}, the solution of extremal problem \eqref{app06.-Schrodinger_extermal_problem} can be written as
\begin{gather}
\zzct{=}A^{-1}\bb,
\end{gather}
where $A$ and $\bb$ are the $\zzcCount{\times}\zzcCount$ real matrix and the $\zzcCount$-dimensional real vector with entries
\begin{align}
A_{i,j}{=}&\Re[\scpr{\psi_{\zc_i}}{\psi_{\zc_j}}]+\epsilon D_{i,j},\\
b_j{=}&\Re[\tfrac{-i}{\hbar}\matel{\psi_{\zc_j}}{\hat H}{\psi}{-}\scpr{\psi_{\zc_j}}{\psi_t}].
\end{align}
The role of the regularization matrix $D$ is twofold. First, it guarantees that $A>0$ and, hence, is invertible for any $\epsilon{>}0$. Second, it introduces a bias for solutions minimizing norms $|\zct_{k}|$ in degenerate cases. The latter feature is especially useful when $\zzc$ are slow-varying parameters. The empirical constant $\epsilon$ should be chosen large enough to provide an effective regularization, but small enough not to influence the computational accuracy.

The specific definitions of parameters $\zzc$ are identical for both the CCS and phase-space-flows-based computations. Namely, $\zzc{=}\{\bc_1,...,\bc_{\dimensionality},\jc_1,...,\jc_{\dimensionality}\}$, where real parameters $\bc_k$ and $\jc_k$ are the amplitude and ``slow'' phase of the complex amplitudes $\ac_k$:
\begin{gather}
\ac_k(t){=}\bc_k(t)e^{i(\jc_k(t){+}\jFastc_k(t))}.
\end{gather}
Here the ``fast'' phase $\jFastc_k(t)$ is defined following the standard practice (see, e.g., Ref.~\cite{2017-Makhov}) by initial condition
$\jFastc_k(0){=}0$ and the evolution equation $\tpder{}{t}\jFastc_k(t){=}\tfrac1{\hbar}(\ppc_k{\cdot}{\dot{\xxc}_k}{-}\matel{\ppc_{k},\xxc_{k}}{\hat H}{\ppc_{k},\xxc_{k}})$.


In all the simulations (except for CCS) we utilized the generalized (''thermal'') Husimi transform by modifying Eq.~\eqref{10.-HF_definition} as
\begin{gather}\label{app06.-thermal_HF_definition}
\HF(\pp,\xx){=}\KernWD_{\alpha\frac{\hbar}{\sqrt2\wwc},\alpha\frac{\wwc}{\sqrt2}}{\cdot}\WF(\pp,\xx)
\end{gather}
using the scaling factor $\alpha{=}1.1$.
%
%
The initial wavefunction $\psi(t{=}0)$ was expanded to 21 basis squeezed coherent states with $\wc_k{=}\frac12$. The expansion coefficients at $t{=}0$ are summarized in Table~\ref{@TBL:num_scs_basis}. 

The source code used to produce all the numerical examples can be found in Ref.~\cite{SOURCE_CODE}.